\documentclass[10pt,a4paper]{article}
\usepackage{amssymb}
\usepackage{amsmath}
\usepackage{amsfonts}
\usepackage{amsthm}
\usepackage{latexsym}
\usepackage{mathrsfs}
\usepackage{stmaryrd}
\usepackage[dvips]{epsfig}
\usepackage{setspace}
\usepackage{float}
\usepackage{bm}
\usepackage{tikz}
\usepackage{enumerate}
\usepackage{subfigure}
\usepackage{nomencl}
\usepackage{authblk}

\usepackage{textcomp}
\usepackage{scrextend}% add KOMA-Script features to other classes
\usepackage[toc,page]{appendix}
\usepackage{tikz-cd}
\usepackage{mathtools}

\usepackage{comment}

\theoremstyle{plain}
\newtheorem{proposition}{Proposition}
\newtheorem{lemma}{Lemma}
\newtheorem{theorem}{Theorem}

\newtheorem{remark}{Remark}
\newtheorem*{notation*}{Notation}

\def\bmd{{\bm d}}

\def\bmg{{\bm g}}
\def\bmh{{\bm h}}

\def\bmu{{\bm u}}
\def\bmv{{\bm v}}

\def\bmJ{{\bm J}}

\def\bmzero{{\bm 0}}

\def\bmA{{\bm A}}

\def\bmF{{\bm F}}

\def\bmK{{\bm K}}

\def\bmQ{{\bm Q}}

\def\bmS{{\bm S}}
\def\bmT{{\bm T}}
\def\bmX{{\bm X}}

\def\bmgamma{{\bm \gamma}}

\def\bmeta{{\bm \eta}}

\def\bmxi{{\bm \xi}}
\def\bmchi{{\bm \chi}}

\def\bmsigma{{\bm \sigma}}
\def\bmvarsigma{{\bm \varsigma}}

\newtheorem*{theorem*}{Theorem}
\newtheorem*{lemma*}{Lemma}

\newcounter{mnotecount}%[section]

\newcommand{\mnotex}[1]%{}
{\protect{\stepcounter{mnotecount}}$^{\mbox{\footnotesize $\bullet$\themnotecount}}$ 
\marginpar{%\color{red}%
\raggedright\tiny\em
$\!\!\!\!\!\!\,\bullet$\themnotecount: #1} }

\allowdisplaybreaks
\begin{document}

\title{\textbf{A perturbative approach to the construction of initial data on compact manifolds}}

\author[,1]{J.A. Valiente Kroon \footnote{E-mail address:{\tt j.a.valiente-kroon@qmul.ac.uk}}}
\author[,1]{J.L. Williams \footnote{E-mail address:{\tt j.l.williams@qmul.ac.uk}}}
%\author[1]{Con T. Ributor}
\affil[1]{School of Mathematical Sciences, Queen Mary, University of London,
Mile End Road, London E1 4NS, United Kingdom.}

\maketitle

\begin{abstract}
We discuss the implementation, on compact manifolds, of the perturbative method of 
Friedrich-Butscher for the construction of solutions to the
vacuum Einstein constraint equations. This method is of a perturbative
nature and exploits the properties of the extended constraint
equations ---a larger system of equations whose solutions imply a
solution to the Einstein constraints. The method is applied to the
construction of nonlinear perturbations of constant mean curvature
initial data of constant negative sectional curvature. We prove
the existence of a neighbourhood of solutions to the constraint
equations around such initial data, with particular components of the
extrinsic curvature and electric/magnetic parts of the spacetime Weyl
curvature prescribed as free data. The space of such free data is
parametrised explicitly.
\end{abstract}

\section{Introduction}

The problem of constructing initial data for the Cauchy problem in General
Relativity, with origins in the work
of Lichnerowicz, has proven to be a rich and interesting
 problem both from the mathematical and
the physical points of view. Recall that an initial data set for the 
Cauchy problem in General Relativity
 consists of a triple $(\mathcal{S},\bmh,\bmK)$, with $\mathcal{S}$
 a 3-dimensional smooth orientable manifold (the \textit{initial hypersurface}), $\bmh$ a
Riemannian metric on $\mathcal{S}$, and $\bmK$ (the \textit{extrinsic curvature}) a
symmetric 2-tensor over $\mathcal{S}$, satisfying the
\emph{Einstein constraint equations}
\begin{subequations}
\begin{eqnarray}
&& r[\bmh]+K^2-K_{ij}K^{ij}=2\lambda,\label{Eq:HamiltonianConstraint}\\
&& D^i K_{ij}-D_jK=0.\label{Eq:MomentumConstraint}
\end{eqnarray}
\end{subequations}
Here, $r[\bmh]$ denotes the Ricci scalar curvature of $\bmh$ and $K\equiv h^{ij}K_{ij}$, the 
\textit{mean extrinsic curvature}. 
Given a solution to the Einstein constraints, the foundational result of Choquet-Bruhat (see \cite{Cho08})
guarantees the existence of a Cauchy development, $(\mathcal{M},\bmg)$, of 
$(\mathcal{S},\bmh,\bmK)$ ---i.e. a solution $(\mathcal{M},\bmg)$ to the Einstein 
field equations with $\bmh$ and $\bmK$ equal to the first and second 
fundamental forms induced by 
$\mathcal{S}\hookrightarrow\mathcal{M}$.  The \textit{Hamiltonian} and \textit{momentum}
constraints \eqref{Eq:HamiltonianConstraint}--\eqref{Eq:MomentumConstraint} comprise a 
highly-coupled system of partial differential
equations, and their analysis therefore presents a significant challenge. The challenge 
is, however, twofold: in addition to the mathematical 
difficulty of analysing such a system of equations, there is on the other hand the difficulty 
of ensuring that the solutions, however obtained, are \emph{physically meaningful}.
 The latter problem is increasingly pertinent as
we move into the age of gravitational wave astronomy.
\medskip

To date, the most popular solution methods have been the so-called \emph{conformal
method} of Lichnerowicz and Choquet-Bruhat (see
e.g. \cite{Cho08}), and the related \emph{conformal thin sandwich}
method. Additionally, there are various techniques based on ``gluing"
constructions, for example.
% with applications to problems of \textit{gravitational shielding}, 
% for instance ---see []. 
For an overview of these methods, we refer the reader to
\cite{BarIse04,Cho08,GalMiaSch15,Ren08}. These techniques share in common the fact that they
rely on reformulating the constraint equations (which are
underdetermined elliptic) as a system of
elliptic PDEs ---requiring, in particular, the
appropriate choice of \textit{freely prescribed} and
\textit{determined fields}--- to which the tools of the theory
elliptic PDEs may then be applied.
One of the features of the conformal method, in particular, is that the free data
are \textit{York-scaled}, so that one needs to solve the full system of 
(conformally formulated) constraint equations, solving in particular
for the conformal factor, before one can obtain the corresponding 
physically meaningful counterparts of the
free data via conformal rescaling. Recent work aiming at making the conformal
method more physically relevant can be found in
e.g. \cite{Max11,Max14}. 
%We point out that, while the elliptic approach is more standard, and
% perhaps the desirable approach from the point of view of asymptotics,
%there are recent alternative approaches based on the theory of
% symmetric hyperbolic/parabolic systems ---see [Racz].

\medskip
The purpose of the present article is to explore an alternative
\emph{perturbative} approach (to be called the
\emph{Friedrich-Butscher} method), first considered in
\cite{But06,But07} and implemented there to prove the existence of
non-linear perturbative solutions of the constraint equations around
flat initial data. The method was adapted in \cite{Delay09} to prove,
in particular, the existence of constant scalar curvature manifolds as
perturbations of hyperbolic space, and to hence construct
hyperboloidal (umbilical) initial data sets that can be thought of as
perturbations of the standard hyperboloid of Minkoswki space. Here we
will be interested in applications to closed (i.e. compact, without
boundary) initial hypersurfaces $\mathcal{S}$ ---i.e. the construction
of initial data for ``cosmological spacetimes". In this approach, the
central object of study is the system of so-called
\textit{extended constraint equations}. While the extended constraint equations 
are entirely equivalent to the Einstein constraint equations ---see
Section \ref{Sec:ECEs}--- their additional structure naturally lends
itself to a choice of freely prescribed data and determined fields
that differs from that of the conformal method. In particular, in this method
 certain components of the Weyl curvature (restricted to the initial
 hypersurface 
  $\mathcal{S}$) of the development $(\mathcal{M},\bmg)$
have the natural interpretation of being freely prescribed data. Note that 
since the method is not based on a conformal reformulation of the constraints,
the free data are physical in the sense of determining, \emph{a priori},
physically relevant properties of the initial data set.
This method, therefore, offers a new perspective on the classical
problem of identifying the \emph{gravitational degrees of freedom} of solutions
 to the Einstein field equations ---the
free data can be thought as parametrising the space of solutions of the
constraints in a neighbourhood of the given background initial data set. 
Although local, in the sense that the free data is given with reference
to a fixed background solution, this is perhaps a natural approach 
within the framework of the Cauchy problem, in particular in problems relating
to Cauchy stability.
\medskip 

The extended constraint equations can also be seen as
a particular case of the \textit{conformal constraint equations} 
 of Friedrich (see \cite{Fri83}), corresponding to
a trivial conformal factor. The conformal constraint equations offer a promising alternative
for the construction (on non-compact manifolds) of initial data with 
\emph{controlled asymptotics}. A detailed understanding 
of the extended constraints is a necessary first step towards the study 
of the conformal constraint equations. 
\medskip 

In restricting to the case of closed initial hypersurfaces,
$\mathcal{S}$, we hope to bring to the foreground the more geometric
aspects of the method, emphasising the key structural features of the
extended constraints that enable such an approach. In the first half
of the article ---Sections \ref{Sec:ECEs} and
\ref{Sec:ButschersMethod}--- we discuss in fairly general terms the
main aspects of the method, identifying structural features of the
extended constraint equations, in addition to the
potential restrictions imposed on the background initial data. In
particular, we identify certain obstructions to the implementation of the
method, at least in its present form ---see Section
\ref{Section:ObstructionsToSolutions}. As proof of concept, the method
is then implemented for a class of background initial data which we
refer to as \textit{conformally rigid hyperbolic} initial data. Here,
the property of \textit{conformal rigidity} is, roughly speaking, the
requirement that there exist no perturbations of the metric that
preserve \emph{conformal flatness} to first order (except, of course,
the pure-gauge perturbations) ---in the case considered here, this is
equivalent to the requirement that the metric admit no tracefree
\textit{Codazzi tensors}, see Section
\ref{Section:ObstructionsToSolutions} for more details.  Such a
background solution may be thought of as constant extrinsic mean
curvature (\textit{CMC}) initial data for a spatially compact analogue
of the $k=-1$ \textit{Friedmann–-Lema\^{i}tre-–Robertson-–Walker}
spacetime. We will see in Section \ref{Section:ParametrisingFreeData}
that this class of background initial data, being conformally flat,
has the additional feature that it allows for an explicit construction
and parametrisation of the free data.

\medskip  
So far, it is unclear whether the obstructions to the method associated to the
existence of globally defined conformal Killing vectors and
Codazzi tensors are
an unavoidable deficiency of the method, or whether they can be
overcome with some appropriate modifications. An analogy can be drawn
here with the conformal method, in which the existence of a
non-trivial conformal Killing vector for the seed metric is an
obstruction to its
implementation ---see, for example, \cite{BarIse04}. Similar
obstructions also arise in the gluing methods.  In the case of the
conformal method, there have been recent attempts to remove the
assumption of the non-existence of conformal Killing fields ---see,
for example \cite{HolMaxMaz17}.  It is plausible that the obstructions
in the Friedrich--Butscher method, too, are not essential.
% \mnotex{JAVK (27.12.2017): need to say more about what are these
%   conformally rigid initial data sets!}
% \mnotex{JLW:Done, I think.}
\medskip

The main result of this
article can be summarised as follows:

\begin{theorem*}
Let $(\mathcal{S},\mathring{\bmh},\mathring{\bmK})$ be a
conformally rigid hyperbolic initial data set on a compact manifold $\mathcal{S}$. Then for each pair of
sufficiently small tensor fields $T_{ij},\bar{T}_{ij}$ over $\mathcal{S}$, transverse-tracefree
with respect to $\mathring{\bmh}$, and each  sufficiently small scalar field $\phi$ over
$\mathcal{S}$, 
%sufficiently close to $\mathring{K}\equiv\text{tr}_{\mathring{\bmh}}\mathring{\bmK}$,
there exists a solution of the Einstein constraint equations $(\mathcal{S},\bmh,\bmK)$  with
$tr_{\mathring{\bmh}}(\bmK-\mathring{\bmK})=\phi$ and for which the electric and
magnetic parts of the Weyl curvature (restricted to
$\mathcal{S}$) of the
resulting spacetime development take the form 
\begin{align*}
& S_{ij}=\mathring{L}(\bmX)_{ij}+T_{ij}-\tfrac{1}{3}\text{tr}_{\bmh}(\mathring{L}(\bmX)+\bmT)~h_{ij},\\
& \bar{S}_{ij}=\mathring{L}(\bar{\bmX})_{ij}+\bar{T}_{ij}-\tfrac{1}{3}\text{tr}_{\bmh}(\mathring{L}(\bar{\bmX})+\bar{\bmT})~h_{ij},
\end{align*}
for some covectors $\bmX,~\bar{\bmX}$ over $\mathcal{S}$, where $\mathring{L}$ denotes the
\emph{conformal Killing operator} with respect to $\mathring{\bmh}$.
%---see section \ref{Sec:Preliminaries}
%Moreover, the free data $\bmT,\bar{\bmT}$ is expressible (unique up to Killing forms) in the form $T_{ij}=H(\eta)_{ij},~\bar{T}_{ij}=H(\bar{\eta})_{ij}$ for some symmetric $\bmh-$tracefree tensors $\eta_{ij},~\bar{\eta}_{ij}$ ---the free data is parametrised by $H$, the linearised Cotton map. 
\end{theorem*}

A precise statement of the above theorem is given in  Section
\ref{Section:PerturbationsOfCompactHyperbolicInitialData}, Theorem \ref{MainTheorem}.

\subsubsection*{Outline of the article}

The structure of this article is as follows: in Section
\ref{Sec:ECEs} we introduce
the extended constraint equations and discuss their relationship to the
Einstein constraint equations. In Section \ref{Sec:ButschersMethod}, 
we describe in general terms the Friedrich-Butscher method; in Section \ref{Sec:SecondarySystem}
we outline the general procedure for the reformulation of the extended
constraint equations as an elliptic
 system; the potential obstructions to the implementation of the method 
 are discussed in Section \ref{Section:ObstructionsToSolutions}, motivating our subsequent
  restriction to \textit{conformally rigid hyperbolic} background initial data. In
Section \ref{Section:PerturbationsOfCompactHyperbolicInitialData} the
method is carried out in this case, the main result being given in
Theorem \ref{MainTheorem} of Section \ref{Sec:MainResult}, and proved by means of Propositions
\ref{PropExistenceOfSolutionsToSecondarySystem} and
\ref{PropositionSufficiency} in Sections \ref{Section:ExistenceOfSolutions}
and \ref{Section:Sufficiency}. 

\subsubsection*{Notation and Conventions}

In the following we will use $(\mathcal{S},\bmh)$ to denote a
Riemannian manifold. The metric $\bmh$ is assumed to be positive
definite. The Levi-Civita connection will be denoted by $D$, and the
Latin indices $i,j,k,\ldots$ will denote abstract tensorial
3-dimensional 
indices. Where convenient we make use of \emph{index-free} notation in
which tensorial objects are written in boldface.

Our conventions for the Riemann curvature are fixed by 
\[
(D_i D_j-D_j D_i)v^k=r^k{}_{lij}v^l.  
\]
The Ricci curvature and scalar are $r_{ij}\equiv r^l{}_{ilj},~ r\equiv h^{ij}r_{ij}$.

% \section{The extended constraint equations and Butscher's method}
% \label{Section:ExtendedConstraintsAndButschersMethod}

% In this section we provide a brief discussion of the extended
% constraint equations, emphasising some of the structural properties
% that will be important in the remainder of this article, and describe
% the procedure of Butscher for constructing nonlinear perturbations.

\section{The extended Einstein constraint equations}
\label{Sec:ECEs}

\medskip
The \emph{extended Einstein constraint equations} (or \textit{extended
constraints} for short) on a spacelike hypersurface $\mathcal{S}$ of a
4-dimensional Lorentzian manifold $(\mathcal{M},\bmg)$ are given by
the conditions 
\begin{equation}
J_{ijk}=0,\qquad \bar{\Lambda}_i=0,\qquad \Lambda_i=0,\qquad V_{ij}=0, 
\label{ExtendedConstraints}
\end{equation}
in terms of the \textit{zero-quantities}
\begin{subequations}
\begin{eqnarray}
&& J_{ijk} \equiv   D_{i}K_{jk} -  \
D_{j}K_{ik}- \epsilon^{l}{}_{ij} \bar{S}_{kl},\label{ECEExtrCurv}\\
&& \Lambda_{i} \equiv D_{j}S_{i}{}^{j}- \epsilon_{ikl} K^{jk} \bar{S}_{j}{}^{l},\label{ECEElectric}\\ 
&& \bar{\Lambda}_{l} \equiv D^{i}\bar{S}_{il}- \epsilon_{ljk} K_{i}{}^{k}r^{ij}  , \label{ECEMagnetic}\\
&& V_{ij} \equiv r_{ij}  -  \tfrac{2}{3} \lambda h_{ij}- S_{ij} -  K_{i}{}^{k} K_{jk} + K_k{}^k K_{ij} .\label{ECEMetric}
\end{eqnarray}
\end{subequations}
They are to be read as equations for a Riemannian metric $h_{ij}$, a
symmetric $2$-tensor $K_{ij}$ to be interpreted as the extrinsic
curvature, and two symmetric $\bmh$-tracefree tensors
$S_{ij},~\bar{S}_{ij}$.

\medskip
The system \eqref{ECEExtrCurv}-\eqref{ECEMetric} can be seen as a
particular case of Friedrich's conformal constraint equations ---namely, when the conformal rescaling is
trivial, see \cite{CFEBook}. The equations associated to the
zero-quantities 
\eqref{ECEExtrCurv} and \eqref{ECEMetric} are nothing other than the
\emph{Codazzi--Mainardi} and \emph{Gauss--Codazzi equations} ---recall
that in three dimensions the essential components of the Riemann
curvature tensor are contained in the
Ricci tensor. The equations associated to the zero-quantities defined
in \eqref{ECEElectric}-\eqref{ECEMagnetic} are
the projections onto $\mathcal{S}$ of the second Bianchi identity of
the ambient spacetime (assuming that the Einstein vacuum field
equations hold):
\[
\nabla_{[a}C_{bc]de}=0, 
\]
where $C_{abcd}$ denotes the Weyl tensor. 
Accordingly, the fields $S_{ij}$ and $\bar{S}_{ij}$ can be
interpreted, respectively, as the \emph{electric} and \emph{magnetic
  parts} of $C_{abcd}$ with respect to the normal of
$\mathcal{S}$ ---the latter 3-manifold being thought of as a spacelike
hypersurface of a spacetime $(\mathcal{M},\bmg)$. 

\begin{remark}
{\em The equations associated to the zero-quantities defined in
\eqref{ECEElectric}-\eqref{ECEMagnetic} may also be interpreted as
integrability conditions for the equations associated to
\eqref{ECEExtrCurv} and \eqref{ECEMetric}. More specifically, the
zero-quantities satisfy the relations
\begin{subequations}
\begin{eqnarray}
&& \bar{\Lambda}_{l}+\tfrac{1}{2}\epsilon_{ijk} D^{k}J^{ij}{}_{l}=0 ,\label{IntegrabilityConditionExtrCurvOLD}\\
&& \Lambda_{j}+D_{i}V_{j}{}^{i}-\tfrac{1}{2}D_jV_i{}^i -   K_{ik}J_{j}{}^{ik} +
 K_{jk}J^{ik}{}_{i}+ KJ_{j}{}^{i}{}_{i} =D^ir_{ij}-\tfrac{1}{2}D_jr=0, \label{IntegrabilityConditionMetricOLD}
\end{eqnarray}
\end{subequations}
where in the latter we are making use of the contracted Bianchi
identity and $K$ denotes the trace of $K_{ij}$ with respect to
$h_{ij}$. In particular, if $J_{ijk}=0$ and $V_{ij}=0$, then
$\Lambda_i=\bar{\Lambda}_i=0$ automatically.}
\end{remark}

Taking the appropriate traces of \eqref{ECEExtrCurv} and
\eqref{ECEMetric}, one obtains the Einstein constraint equations
\begin{subequations}
\begin{eqnarray}
&& J_{ij}{}^i\equiv D^iK_{ij}-D_j K=0,\\
&& V_i{}^i\equiv r-2\lambda-K_{ij}K^{ij}+K^2=0. \label{HamiltonianConstraint}
\end{eqnarray}
\end{subequations}
It follows then that any solution to the equations associated to the
zero-quantities \eqref{ECEExtrCurv}-\eqref{ECEMetric} gives rise also
to a solution of the Einstein constraints. The reverse is also true,
since, having obtained a solution $(\mathcal{S},\bmh,\bmK)$ of the
Einstein constraints, one simply \emph{defines}
\begin{subequations}
\begin{eqnarray}
&& S_{ij}=r_{ij}  -  \tfrac{2}{3} \lambda h_{ij}-  K_{i}{}^{k} K_{jk} +K K_{ij} ,\label{Eq:ElectricPartFromInitialData}\\
&& \bar{S}_{kl} =- \epsilon_{lij} D^{j}K_{k}{}^{i}.\label{Eq:MagneticPartFormInitialData}
\end{eqnarray}
\end{subequations}
By construction then we have $J_{ijk}=0,~V_{ij}=0$, whence the
integrability conditions imply $\Lambda_i=\bar{\Lambda}_i=0$. Hence,
solutions of the extended constraints and of the Einstein constraint
equations are in direct correspondence.

\begin{remark}\label{Remark:EMVersionIsCoupled}{\em
Note that, assuming $V_{ij}=0$, if one substitutes \eqref{ECEMetric} into \eqref{ECEMagnetic} one obtains
\begin{equation}
 \bar{\Lambda}_{l} \equiv D^{i}\bar{S}_{il} - S^{ij} \epsilon_{ljk} K_{i}{}^{k},\label{EMVersionOfElectricEq}
\end{equation}
which better exhibits the \emph{electromagnetic duality} between the
electric and magnetic parts of the Weyl tensor: namely, that under the transformation 
\[
S_{ij}\longrightarrow \bar{S}_{ij},\qquad \bar{S}_{ij}\longrightarrow
-S_{ij},\]
the corresponding zero quantities transform as 
\[\Lambda_{i}\longrightarrow\bar{\Lambda}_i,\qquad
\bar{\Lambda}_{i}\longrightarrow -\Lambda_i. 
\]
We choose, however, to work with the system
\eqref{ECEExtrCurv}--\eqref{ECEMetric}, since the resulting
integrability conditions
\eqref{IntegrabilityConditionExtrCurv}--\eqref{IntegrabilityConditionMetric}
enjoy a particular semi-decoupling of the zero-quantities $J_{ijk}$
and $V_{ij}$ that is convenient for the subsequent analysis, and that
is lost if one uses the alternative definition of the zero-quantity $\bar{\Lambda}_i$,
given by \eqref{EMVersionOfElectricEq}.}
\end{remark}

\section{The Friedrich--Butscher Method}
\label{Sec:ButschersMethod}
In this section, we outline the general procedure introduced in
\cite{But06,But07} to construct solutions to the Einstein constraint
equations, in addition to describing some of the potential
obstructions to its implementation. As mentioned in the introduction,
the procedure is of a perturbative nature ---that is,  one proves the
existence of nonlinear perturbations of some \emph{background initial
data set}, denoted $(\mathcal{S},\mathring{\bmh},\mathring{\bmK})$,
 through the use of the \emph{implicit function theorem}. In
order to apply the implicit function theorem, one first derives from
the extended constraint equations a so-called
\textit{auxiliary}
system of equations which, given the appropriate choice of free and
determined data, has a linearisation which is manifestly elliptic. By
construction, any solution of the extended constraint equations is
also a solution of the auxiliary equations. Having found, via the
inverse function theorem, 
an open neighbourhood of solutions to the auxiliary system around the
given background initial data set one must then show that such
\emph{candidate initial data set} is indeed a solution to the extended
constraints ---we refer to the latter as the problem of
\emph{sufficiency of the auxiliary system}.

\medskip
 In short, the Friedrich--Butscher method may be divided into two stages:
\begin{enumerate}
\item[(i)] \textbf{\em Construction of candidate solutions:} derive a
auxiliary system of equations, with elliptic linearisation, and apply
the implicit function theorem to guarantee existence of solutions.
\item[(ii)] \textbf{\em Sufficiency}: prove that the solutions to the
auxiliary system constructed in \emph{Step (i)} are also solutions to the
extended constraint equations.
\end{enumerate}

In Section \ref{Section:ObstructionsToSolutions} we discuss the
potential obstructions to the implementation of the above procedure. The
desire to avoid such obstructions motivates our restriction to
\textit{conformally rigid hyperbolic} manifolds in Section
\ref{Section:PerturbationsOfCompactHyperbolicInitialData}.

\subsection{Preliminaries}
\label{Sec:Preliminaries}
In the following, it will be convenient to a adopt a slightly more
index-free notation that emphasises the structure of the
equations. Given the Riemannian 3-manifold $(\mathcal{S},\bmh)$, we
introduce the following spaces of tensors:
\begin{itemize}
\item $\Lambda^1(\mathcal{S})$, the space of covectors over $\mathcal{S}$;
\item $\mathscr{S}^2(\mathcal{S})$, the space of symmetric $2$-tensors
  over $\mathcal{S}$;
\item $\mathscr{S}^2_0(\mathcal{S};\bmh)$, the space of symmetric
$2$-tensors over $\mathcal{S}$ that are tracefree with respect to the
metric $\bmh$;
\item $\mathscr{S}_{TT}(\mathcal{S};\bmh)$, the space of
  transverse-tracefree tensors over $\mathcal{S}$ with respect to the
  metric $\bmh$;
\item $\mathcal{J}(\mathcal{S})$, the space of \textit{Jacobi} tensors
  ---i.e. tensors $J_{ijk}$ satisfying 
\[
J_{ijk}=-J_{jik},\qquad J_{ijk}+J_{jki}+J_{kij}=0 .
\] 
\end{itemize}

\begin{remark}\label{Remark:JacobiDecomp}
{\em It will be useful to note that 
\[
\mathcal{J}(\mathcal{S})\simeq \Lambda^1(\mathcal{S})\oplus
\mathscr{S}^2_0(\mathcal{S};\bmh).
\]
More precisely, any $J_{ijk}\in\mathcal{J}(\mathcal{S})$ may be uniquely decomposed as 
\begin{equation} 
\label{DecompJacobiTensor}
J_{ijk} = \tfrac{1}{2} \left( \epsilon_{ij}{}^{l}F_{lk} + A_{i} h_{jk}-  A_{j} h_{ik} \right),
\end{equation}
where
\[
A_{j} \equiv J_{jk}{}^k,\qquad F_{km} \equiv
\epsilon_{ij(m}J^{ij}{}_{k)}, 
\]
the latter being tracefree. In the previous expressions and in the following
$\epsilon_{ijk}$ denotes the \emph{volume form} of the metric $\bmh$.
We will refer to \eqref{DecompJacobiTensor} as the \emph{Jacobi decomposition, with respect to} $\bmh$ of $J_{ijk}$.}
\end{remark}

We also introduce the following operators:
\begin{itemize}
\item
  $\Pi_\bmh:\mathscr{S}^2(\mathcal{S})\longrightarrow\mathscr{S}^2_0(\mathcal{S};\bmh)$,
  the \emph{projection of symmetric 2-tensors into the space of symmetric
  tracefree 2-tensors}, given by 
\[
\Pi_\bmh(\eta)_{ij}\equiv \eta_{ij}-\tfrac{1}{3}\text{tr}_\bmh(\bmeta)
h_{ij};
\]
\item
  $\star:\mathscr{S}^2_0(\mathcal{S};\bmh)\longrightarrow\mathcal{J}(\mathcal{S})$, given by
\[
(\star \eta)_{ijk}\equiv \epsilon^l{}_{ij}\eta_{kl};
\]
where $\epsilon_{ijk}$ denotes the volume form;
\item
  $\delta_\bmh:\mathscr{S}^2(\mathcal{S})\longrightarrow\Lambda^1(\mathcal{S})$,
  the \emph{divergence operator}, 
\[
\delta_\bmh(\eta)_j\equiv D^i\eta_{ij};
\]
\item $ L_\bmh:\Lambda^1(\mathcal{S})\longrightarrow \mathscr{S}^2_0(\mathcal{S};\bmh)$ the \emph{conformal Killing operator}, 
\[
L_\bmh(X)_{ij}\equiv D_i X_j+D_j X_i-\tfrac{2}{3}D^k X_k h_{ij};
\]
\item $\mathcal{D}_\bmh: \mathscr{S}^2(\mathcal{S})\longrightarrow\mathcal{J}(\mathcal{S})$ the \emph{Codazzi operator}, 
\[
\mathcal{D}_\bmh(\eta)_{ijk}\equiv D_i\eta_{jk}-D_j\eta_{ik},
\]
\item
  $\mathcal{D}_\bmh^*:\mathcal{J}(\mathcal{S})\longrightarrow\mathscr{S}^2(\mathcal{S})$,
the \emph{formal $L^2$-adjoint of $\mathcal{D}_\bmh$ restricted to
$\mathscr{S}^2_0(\mathcal{S};\bmh)$}, and given by
\[
\mathcal{D}_\bmh^*(\mu)_{ij}\equiv D^k \mu_{ikj}+D^k \mu_{jki}-\tfrac{2}{3}D^k
\mu_{lk}{}^l h_{ij};
\]
\item $\Delta_L: \mathscr{S}^2(\mathcal{S})\longrightarrow\mathscr{S}^2(\mathcal{S})$, the \textit{Lichnerowicz Laplacian}, acting as  
\[
\Delta_L \eta_{ij}\equiv -\Delta_h \eta_{ij}+2 r_{(i}{}^k
\eta_{j)k}-2r_{ikjl}\eta^{kl},
\]
where $\Delta_\bmh\equiv h^{ij}D_i D_j$ is the \emph{rough
  Laplacian}. 
%  Recall that $\Delta_L$ is the result of composing
%$D_{\text{sym}}$ with its formal $L^2-$adjoint ---here,
%$(D_{\text{sym}}\eta)_{ijk}\equiv D_{(i}\eta_{jk)}$.
\end{itemize}

\begin{notation*}\em{Often, for the sake of simplicity, the subscript
    $\bmh$ in the symbol of the above operators will be omitted. When
    the above operators are defined with respect to the background
    metric $\mathring{\bmh}$ they will be distinguished by the symbol $\mathring{~}$. }
\end{notation*}

\begin{remark}\label{Remark:JacobiDecompOfCodazziOperator}{\em Since $\mathcal{D}_\bmh: \mathscr{S}^2(\mathcal{S})\longrightarrow\mathcal{J}(\mathcal{S})$, the image of $\mathcal{D}_\bmh$ may be decomposed as in Remark \ref{Remark:JacobiDecomp}. In particular, given $\eta_{ij}\in \mathscr{S}^2_0(\mathcal{S};\bmh)$, $\mathcal{D}_\bmh (\bm\eta)_{ijk}$ may be decomposed as follows
\begin{equation}
\mathcal{D}_\bmh (\bm\eta)_{ijk}=\tfrac{1}{2}(\epsilon_{ij}{}^l\text{rot}_2(\bm\eta)_{lk}-\delta_\bmh(\bm\eta)_ih_{jk}+\delta_\bmh(\bm\eta)_jh_{ik}),
\end{equation}
where $\text{rot}_2(\bm\eta)_{ij}\equiv \epsilon_{kl(i}D^k\eta^l{}_{j)}$. It therefore follows that $\mathcal{D}_\bmh (\bm\eta)_{ijk}=0$ for $\eta_{ij}\in \mathscr{S}^2_0(\mathcal{S};\bmh)$ if and only if $\delta(\bm\eta)_i=0$ \emph{and} $\text{rot}_2(\bm\eta)_{ij}=0$.
}
\end{remark}

We recall  that the divergence operator is undetermined
elliptic and (equivalently) the conformal Killing operator $L$ is
overdetermined elliptic. Moreover, as shown in \cite{But07}, the
operator $\mathcal{D}_{\bmh}$ is overdetermined elliptic when
restricted to $\mathscr{S}^2_0(\mathcal{S};\bmh)$. More precisely, one
has the following:

\begin{lemma}
Given a covector $\bmxi$ let 
\[
\sigma_\bmxi[\mathcal{D}_\bmh]: \mathscr{S}^2(\mathcal{S})\longrightarrow
\mathcal{J}(\mathcal{S})
\]
 denote the symbol map of
$\mathcal{D}_{\bmh}$. For $\bmxi\neq 0$, the kernel of
$\sigma_\bmxi[\mathcal{D}_\bmh]$ is one dimensional ---it consists of
elements of the form $c\xi_i\xi_j$. It follows that the operator
$\mathcal{D}_{\bmh}\vert_{\mathscr{S}^2_0(\mathcal{S};\bmh)}$ is
overdetermined elliptic.
\end{lemma}

The proof is straightforward; the details can be found in \cite{But07}.

\begin{remark}\label{Remark:DecompositionOfCodazziOperator}
{\em In terms of the above definitions, the extended constraints encoded in
the zero-quantities \eqref{ECEExtrCurv}--\eqref{ECEMetric} may be rewritten as 
\begin{subequations}
\begin{eqnarray}
&& \mathcal{D}_{\bmh}(K)_{ijk}-(\star\bar{S})_{ijk}=0, \label{ExtrinsicCurvatureAuxiliary}\\
&& \delta_{\bmh} (S)_{i}+\epsilon^{jk}{}_i K_j{}^l \bar{S}_{kl}=0,\\
&& \delta_{\bmh}(\bar{S})_{i}- \epsilon_{i}{}^{jk} K_{k}{}^{l}r_{lj}=0,\label{}\\
&& r_{ij}-\tfrac{2}{3}\lambda h_{ij}-S_{ij}+K K_{ij}-K_i{}^k K_{jk}=0.
\end{eqnarray}
\end{subequations}}
\end{remark}

%\begin{proof}
%Fix $\xi\neq 0$, and  
%\[\sigma_\xi[\mathcal{D}_g](\eta)_{ijk}=\xi_i\eta_{jk}-\xi_j\eta_{ik} \]
%\end{proof}

\subsection{The auxiliary system}
\label{Sec:SecondarySystem}

The Friedrich--Butscher method for the construction of solutions to the Einstein constraint
equations relies on first using the extended constraint
equations to obtain a auxiliary system of equations whose
linearisation is elliptic. The existence of solutions  is then
established through an application of the implicit function
theorem. In general, the linearised system is a highly coupled second
order system of partial differential equations. In the case of
background data with metric of constant sectional curvature
(i.e. Einstein manifolds), the linearised
equations decouple sufficiently so as to enable a straightforward
analysis of its kernel and cokernel ---this system will be given in
Section \ref{Section:ExistenceOfSolutions}. Here, we discuss the
procedure in full generality, but for simplicity we restrict attention
to the principal parts of the equations, since they suffice for the
description of ellipticity. 

\subsubsection{The ansatz}

First note that, given a background initial data set $(\mathcal{S},\mathring{\bmh},\mathring{\bmK})$, there exists (see \eqref{Eq:ElectricPartFromInitialData} and \eqref{Eq:MagneticPartFormInitialData}) a corresponding background solution to the extended constraints, denoted $(\mathcal{S},\mathring{\bmK},\mathring{\bar{\bmS}},\mathring{\bmS},\mathring{\bmh})$, and which may moreover be decomposed as follows
\begin{subequations}
\begin{eqnarray}
&& \mathring{K}_{ij}=\kappa_{ij}+\tfrac{1}{3}\mathring{K}\mathring{h}_{ij},\\
&& \mathring{S}_{ij}=\mathring{L}(\bmv)_{ij}+\psi_{ij},\label{YorkDecompElectric}\\
&& \mathring{\bar{S}}_{ij}=\mathring{L}(\bar{\bmv})_{ij}+\bar{\psi}_{ij},\label{YorkDecompMagnetic}
\end{eqnarray}
\end{subequations}
with $\kappa_{ij}\in\mathscr{S}^2_0(\mathcal{S};\mathring{\bmh})$, $v_i,\bar{v}_i\in\Lambda^1(\mathcal{S})$ and $\psi_{ij},\bar{\psi}_{ij}\in\mathscr{S}_{TT}(\mathcal{S};\mathring{\bmh})$. Decompositions \eqref{YorkDecompElectric} and \eqref{YorkDecompMagnetic} are precisely the \emph{York splits} (see \cite{Yor73,Can81}) of the electric and magnetic parts; such a split is always possible, and is moreover unique up to the addition of conformal Killing fields to $v_i,\bar{v}_i$.
\begin{remark}
{\em In Section \ref{Section:PerturbationsOfCompactHyperbolicInitialData}, we will restrict to background initial data which is Einstein and umbilical, for which $\kappa_{ij}=0$, $v_i=\bar{v}_i=0$ and $\psi_{ij}=\bar{\psi}_{ij}=0$.  }
\end{remark}
We will seek solutions of the extended constraints of the form
\begin{subequations}
\begin{eqnarray}
&& K_{ij}=\kappa_{ij}+\chi_{ij} +\tfrac{1}{3}(\mathring{K}+\phi)~ \mathring{h}_{ij},\label{DecompExtrinsicCurv}\\
&& S_{ij}=\Pi_{\bmh}(\mathring{L}(\bmv+\bmX)+\bm\psi+\bmT)_{ij},\label{ModifiedYorkAnsatzElectric}\\
&& \bar{S}_{ij}=\Pi_{\bmh}(\mathring{L}(\bar{\bmv}+\bar{\bmX})+\bar{\bm\psi}+\bar{\bmT})_{ij}, \label{ModifiedYorkAnsatzMagnetic}
\end{eqnarray}
\end{subequations}
where $\chi_{ij}$ is tracefree with respect to the background metric
$\mathring{\bmh}$, $\mathring{K}+\phi$ being the trace part, and where
$T_{ij},~\bar{T}_{ij}$ are taken to be transverse-tracefree with respect to the background
metric. Recall that $\Pi_{\bmh}$ is the projection onto
$\mathscr{S}^2_0(\mathcal{S};\bmh)$, so that $S_{ij}$ and $\bar{S}_{ij}$
are $\bmh$-tracefree, as required. We will use
$\bmS(\bmX,\bmT),~\bar{\bmS}(\bar{\bmX},\bar{\bmT})$ as shorthands for \eqref{ModifiedYorkAnsatzElectric} and \eqref{ModifiedYorkAnsatzMagnetic}. The above ansatz is motivated by the fact that the operator $\delta_{\bmh}$ is undetermined
elliptic, while
$\mathcal{D}_{\bmh}\vert_{\mathscr{S}^2_0(\mathcal{S};\bmh)}$ is
overdetermined elliptic. Note that the background solution corresponds to taking
\[(\bm\chi,\bar{\bmX},\bmX,\bmh)=(\bm0,\bm0,\bm0,\mathring{\bmh})\qquad \text{and}\qquad (\phi,\bar{\bmT},\bmT)=(0,\bm0,\bm0) \]
in \eqref{DecompExtrinsicCurv}--\eqref{ModifiedYorkAnsatzMagnetic}.

\begin{remark}{\em
% The ansatz
% \eqref{ModifiedYorkAnsatzElectric}-\eqref{ModifiedYorkAnsatzMagnetic}
% is the projection of the York split with respect to the background
% metric, $\mathring{\bmh}$ ---see \cite{Yor73,Can81}. 
Here we adopt a slightly different approach to that of \cite{But06,But07}, which uses the 
 ansatz 
\[
S_{ij}=L_{\bmh}(\bmX)_{ij}+\Pi_{\bmh} T_{ij}, 
\]
with $T_{ij}$ a transverse-tracefree tensor with respect to $\mathring{\bmh}$.
The reason for using 
\eqref{ModifiedYorkAnsatzElectric}--\eqref{ModifiedYorkAnsatzMagnetic}
is that we will be able to use the orthogonality property of the York
split (with respect to $\mathring{\bmh}$) ---see \cite{Can81}--- to
argue, in a straightforward way, that the solutions are uniquely determined by the
freely-prescribed data $(\phi,\bmT,\bar{\bmT})$. }
\end{remark}

\subsubsection{The linearisation of the Ricci operator}

Let us now consider equation \eqref{ECEMetric}. As is well known, the
linearised Ricci operator is not elliptic. The failure of the
linearised Ricci operator to be elliptic is a consequence of
diffeomorphism-invariance, as encoded by the contracted Bianchi
identity ---see, for instance, \cite{ChoKno04}. One method of
breaking the gauge-invariance is via the use of a variation of the so-called
\emph{DeTurck trick}. Here we follow this approach.

\medskip
Let $\mathring{D}$ denote the Levi-Civita connection associated to
$\mathring{\bmh}$. The \emph{linearisation of the Ricci operator}, $\breve{r}(\gamma)_{ij}$, about
$\mathring{h}_{ij}$ acting on a symmetric tensor field $\gamma_{ij}$
(the \textit{metric perturbation}) is given by the following \emph{Fr\'{e}chet derivative}
\begin{align}
\breve{r}(\gamma)_{ij}&\equiv \frac{d}{d\tau}r[\mathring{\bmh}+\tau\bm\gamma]_{ij}\bigg\vert_{\tau=0}\\
&=-\tfrac{1}{2}\mathring{\Delta}\gamma_{ij}+\tfrac{1}{2}\mathring{D}_k\mathring{D}_i\gamma_j{}^k+\tfrac{1}{2}\mathring{D}_k\mathring{D}_j\gamma_i{}^k-\tfrac{1}{2}\mathring{D}_i\mathring{D}_j\gamma\nonumber\\
&=-\tfrac{1}{2}\mathring{\Delta}\gamma_{ij}+\tfrac{1}{2}\mathring{D}_i\mathring{D}_k\gamma_j{}^k+\tfrac{1}{2}\mathring{D}_j\mathring{D}_k\gamma_i{}^k-\tfrac{1}{2}\mathring{D}_i\mathring{D}_j\gamma+\mathring{r}_{(i}{}^k\gamma_{j)k}-\mathring{r}_{ikjl}\gamma^{kl}\nonumber\\
&=\tfrac{1}{2}\Delta_L\gamma_{ij}+\mathring{D}_{(i}C(\bm\gamma)_{j)}{}^k{}_k,\label{LinearisedRicci}
\end{align}
where, here, $\tau$ is a real parameter describing a
a one-parameter-family of metrics,
$\bmh(\tau)=\mathring{\bmh}+\tau\bm\gamma$, and $C(\cdot)^i{}_{jk}$ is
defined by
\begin{equation}
\label{PerturbationOfGradient}
C(\bm\gamma)^i{}_{jk}\equiv \tfrac{1}{2}(\mathring{D}_j\gamma_k{}^i+\mathring{D}_k\gamma_j{}^i-\mathring{D}^i\gamma_{jk}).
\end{equation}
Here, and it what follows, index raising and lowering within a
linearised covariant  will be carried out with respect to the background
metric, $\mathring{\bmh}$. The first term of \eqref{LinearisedRicci},
$\Delta_L\gamma_{ij}$, is manifestly elliptic, but the ellipticity is
spoiled by the second-order term
$\mathring{D}_{(i}C_{j)}{}^k{}_k$. Now, given an arbitrary local coordinate system, $(x^\alpha)$,
define the following
\[
Q(\tau)^\alpha\equiv \frac{1}{2}h(\tau)^{\beta\gamma}(\Gamma(\bmh(\tau))^\alpha_{\beta\gamma}-\mathring{\Gamma}^\alpha_{\beta\gamma}),
\]
where $h(\tau)^{\beta\gamma}$ is the inverse of
$h(\tau)_{\alpha\beta}$, and
$\Gamma(\bmh(\tau))^\alpha_{\beta\gamma}$,
$\mathring{\Gamma}^\alpha_{\beta\gamma}$ denote respectively the
Christoffel symbols of the metrics $\bmh(\tau)$ and $\mathring{\bmh}$
in the local coordinates, $(x^\alpha)$.

\begin{remark}
\label{Remark:QisAVector}
{\em Note that, though $Q^\alpha$ is defined with respect to a fixed local coordinate system, the expression is in fact covariant,
being given by the trace of the difference of two connections (i.e the trace of the \emph{transition tensor}, $S^k{}_{ij}$). Hence, $\bmQ$ represents a (globally-defined) vector field, which we will denote in the abstract index formalism by $Q^i$. The remaining
calculations of the article will be carried out in the abstract index notation.}
\end{remark}

Consider now the Lie derivative of the metric along $Q(\tau)$, $\mathcal{L}_{Q(\tau)} h(\tau)_{ij}$,
the linearisation of which is given by
\[
\frac{d}{d\tau}(\mathcal{L}_{Q(\tau)}
\bmh(\tau))_{ij}\bigg\vert_{\tau=0} =\mathring{D}_{(i}C_{j)}{}^k{}_k,
\]
which is precisely the term in \eqref{LinearisedRicci} obstructing the
ellipticity in the linearised Ricci operator. Accordingly, we define
the \textit{reduced Ricci operator}, $\text{Ric}^Q(\cdot )$, as
\[
\text{Ric}^Q(\bmh)_{ij}\equiv r_{ij}-(\mathcal{L}_Q h)_{ij}. 
\]
The linearisation of the reduced Ricci operator  can then be seen to
be proportional to the Lichnerowicz Laplacian of the background metric
---that is,
\[
D\text{Ric}^Q(\mathring{\bmh})\cdot
\gamma_{ij}=\tfrac{1}{2}\mathring{\Delta}_L\gamma_{ij} ,
\]
which is manifestly elliptic ---note that, modulo curvature terms,
$\Delta_L$ is simply the rough Laplacian and, therefore, clearly
elliptic ---see e.g. also \cite{DeT80} for an alternative discussion of the above.

\begin{remark}
{\em The reduced Ricci operator coincides with the Ricci operator when
$Q^i=0$. The linearisation $D\text{Ric}^Q(\cdot)$ is formally identical to that
obtained through the use of (generalised) harmonic
coordinates. }
\end{remark}

\subsubsection{The auxiliary extended constraint map}
Following the discussion of the previous subsections, it is convenient
to define the \textit{auxiliary extended constraint map} 
\[
\Psi(\bm\chi,\bar{\bmX},\bmX,\bmh;\phi,\bar{\bmT},\bmT)\equiv\left(\begin{array}{c}
\mathring{\mathcal{D}}^*(\bmJ)_{ij}\\
\bar{\Lambda}_i\\
\Lambda_i\\
V_{ij}-\mathcal{L}_Qh_{ij}
\end{array}\right)
=
\left(\begin{array}{l}
\mathring{\mathcal{D}}^*\left(\mathcal{D}_\bmh(\bmK)-\star\bar{\bmS}\right)_{ij}\\[0.5em]
\delta_\bmh (\bar{\bmS})_{i}-\epsilon^{jk}{}_i\chi_j{}^l S_{kl}\\[0.5em]
\delta_\bmh (\bmS)_{i}+\epsilon^{jk}{}_i\chi_j{}^l \bar{S}_{kl}\\[0.5em]
\text{Ric}^Q(\bmh)_{ij}-\tfrac{2}{3}\lambda h_{ij}-S_{ij}+KK_{ij}-K_i{}^kK_{jk}
\end{array}\right)
\]
with the understanding that the fields $K_{ij},~S_{ij},~\bar{S}_{ij}$
should be substituted by the ansatz
\eqref{DecompExtrinsicCurv}--\eqref{ModifiedYorkAnsatzMagnetic}. In
terms of the latter, the \emph{auxiliary system} is then given by 
\begin{equation}
\Psi(\bmchi,\bar{\bmX},\bmX,\bmh;\phi,\bar{\bmT},\bmT)=0,
\label{AuxiliarySystem}
\end{equation}
which is to be read as a (second-order) system of partial differential
equations for the fields $\bmchi,\bar{\bmX},\bmX,\bmh$ while the
fields $\phi,\bar{\bmT},\bmT$ are regarded as input ---i.e. they are the freely
specifiable data. 

\begin{remark}
{\em Note that the auxiliary system is defined always with reference
to some fixed \textit{background solution} $(\mathring{\bmK},\mathring{\bar{\bmS}},\mathring{\bmS},\mathring{\bmh})$ of the extended constraints, both through the ansatz
\eqref{DecompExtrinsicCurv}-\eqref{ModifiedYorkAnsatzMagnetic} and
through the definition of the reduced Ricci operator. It is straightforward to see that, for any given background solution, we have 
\[\Psi(\bm0,\bm0,\bm0,\bm0;0,\bm0,\bm0)=0 \]
---that is to say, that the background solution (corresponding to trivial free and determined fields) itself solves the corresponding auxiliary equations.
 }
\end{remark}

In the following, we denote by
$D\Psi[\mathring{\bmK},\mathring{\bar{\bmX}},\mathring{\bmX},\mathring{\bmh}]\cdot(\bmsigma,\bar{\bmxi},\bmxi,\bmgamma)$
the linearisation of $\Psi$ at
$(\mathring{\bmK},\mathring{\bar{\bmX}},\mathring{\bmX},\mathring{\bmh})$
in the direction of the determined fields ---that is to say, the
following
\begin{align*}
D\Psi[\mathring{\bmK},\mathring{\bar{\bmX}},\mathring{\bmX},\mathring{\bmh}]\cdot(\bmgamma,\bmsigma,\bmxi,\bar{\bmxi})&=\frac{d}{d\tau}\Psi(\mathring{\bmh}+\tau\bmgamma,\mathring{\bmchi}+\tau\bmsigma,\mathring{\bmX}+\tau\bmxi,\mathring{\bar{\bmX}}+\tau\bar{\bmxi};~\phi,\bar{\bmT},\bmT) \bigg\vert_{\tau=0},  
\end{align*}
where $\mathring{\bmX}_i,~\mathring{\bar{\bmX}}$ are the covector fields appearing in the York decomposition of the background electric and magnetic Weyl curvatures, $\mathring{\bmS},~\mathring{\bar{\bmS}}$, and $\mathring{\bm\chi}$ is the tracefree part of $\mathring{\bmK}$ with respect to $\mathring{\bmh}$. 
%\[\mathring{S}_{ij}-\mathring{L}(\mathring{\bmX})_{ij}\in\mathscr{S}_{TT}(\mathcal{S};\mathring{\bmh}),\qquad \mathring{\bar{S}}_{ij}-\mathring{L}(\mathring{\bar{\bmX}})_{ij}\in\mathscr{S}_{TT}(\mathcal{S};\mathring{\bmh}) \]
%---i.e. the vector fields (unique up to addition of a conformal Killing vector ---see \cite{Can81}) appearing %in the York Ansatz. 
\begin{notation*}
\em{We will often denote
$D\Psi[\mathring{\bmK},\mathring{\bar{\bmX}},\mathring{\bmX},\mathring{\bmh}]$
by $D\Psi$ for notational convenience. }
\end{notation*}

Note that, as they are held fixed, the free data
$(\phi,\bar{\bmT},\bmT)$ are not an input for $D\Psi$. We will not give
the expression for $D\Psi$ for a general background here.  It will
suffice for the purposes of this section to consider only the
principal parts as a second-order system of partial differential
equations ---namely,
\[
\left(\begin{array}{cccc}
\mathring{\mathcal{D}}^*\circ\mathring{\mathcal{D}} & \mathring{\mathcal{D}}^*(\mathring{\star} \mathring{L}) & 0 & 0\\
0 & \mathring{\delta}\circ \mathring{L} & 0 & 0\\
0 & 0 & \mathring{\delta}\circ \mathring{L} & 0\\
0 & 0 & 0 & -\tfrac{1}{2}\mathring{\Delta}
\end{array}\right)\left(\begin{array}{c}
\sigma_{ij}\\
\bar{\xi}_i\\
\xi_i\\
\gamma_{ij}
\end{array}\right).
  \]
Since the principal part is upper-triangular, to verify ellipticity of
the full system we need consider only the diagonal entries, which are
elliptic by construction ---one proceeds from the bottom-right,
verifying invertibility of the symbol of each row, and successively
substituting into the row above where necessary. It follows then that
$D\Psi$ is a Fredholm operator.

\subsection{The sufficiency argument}
\label{Subsection:SufficiencyArgument}

Let us now assume that \emph{Step (i)} (see beginning of Section \ref{Sec:ButschersMethod}) has
been carried out: that is to say, that we have established the
existence of a small neighbourhood of solutions to the auxiliary
system \eqref{AuxiliarySystem}. In particular we have 
\begin{subequations}
\begin{eqnarray}
&& \mathring{\mathcal{D}}^*(J)_{ij}=0,\label{HavingSolvedExtrinsicCurvAuxiliaryEq}\\
&& V_{ij}= (\mathcal{L}_Q \bmh)_{ij},\\
&& \Lambda_i=\bar{\Lambda}_i=0 \label{HavingSolvedElectricandMagneticEquations}.
\end{eqnarray}
\end{subequations}

In order to conclude that such solutions of the auxiliary system
indeed solve the extended constraint equations, there remains the task
of showing:
\begin{enumerate}[(a)]
\item that $(\mathcal{L}_Q \bmh)_{ij}=0$ in order that $\text{Ric}(\bmh)=\text{Ric}^Q(\bmh)$, implying \eqref{ECEMetric};\label{GaugePropagationStatement}
\item that $J_{ijk}=0$ so that \eqref{ECEExtrCurv} is satisfied.  \label{SufficiencyExtrCurvStatement}
\end{enumerate} 

\begin{remark}
{\em Item (\ref{GaugePropagationStatement}) can be thought of as the
analogue of gauge propagation in the hyperbolic reduction of the
Einstein field equations.}
\end{remark}

The tasks (a)-(b) will be established with the help of the
integrability conditions
\eqref{IntegrabilityConditionExtrCurvOLD}-\eqref{IntegrabilityConditionMetricOLD},
which in view of \eqref{HavingSolvedElectricandMagneticEquations},  reduce to 
\begin{subequations}
\begin{eqnarray}
&& \epsilon^{ijk}D_iJ_{jkl}=0, \label{IntegrabilityConditionExtrCurv}\\
&& D^{i}(\mathcal{L}_Q \bmh)_{ij}-\tfrac{1}{2}D_j(\mathcal{L}_Q \bmh)_i{}^i =  K_{ik} J_{j}{}^{ik} -
K_{jk}J^{ik}{}_{i}  - K J_{j}{}^{i}{}_{i}.  \label{IntegrabilityConditionMetric}
\end{eqnarray}
\end{subequations}
The strategy will be to use
\eqref{HavingSolvedExtrinsicCurvAuxiliaryEq} and
\eqref{IntegrabilityConditionExtrCurv} to first show that $J_{ijk}=0$,
and then to substitute into \eqref{IntegrabilityConditionMetric},
which will be used to show $Q_i=0$.

\subsubsection{Elliptic equations for $Q_i$ and $J_{ijk}$}
\label{Subsubsection:EllipticEquationsForQandJ}

%  
%In the forthcoming sections we outline the derivation of two
%integral identities, \eqref{SufficiencyIntegralForQ} and
%\eqref{SufficiencyIntegral}, that will form the basis of the
%sufficiency argument. We follow \cite{But06,But07} ---while the
%derivations are fundamentally the same, we take care to keep track of
%the various curvature quantities that arise, the detailed knowledge of
%which was not required for the purposes of \cite{But06,But07}. On the
%other hand, since we are dealing here with a compact manifold
%$\mathcal{S}$, there are no boundary terms to keep track of when
%integrating-by-parts ---this is not the case in \cite{But06,But07}.
%\\

First, it will prove convenient to first define the operator
\[\mathcal{K}_{\bmh}:\mathcal{J}(\mathcal{S})\longrightarrow \mathscr{S}^2_0(\mathcal{S};\mathring{\bmh})\oplus\Lambda^1(\mathcal{S})\]
acting as
\[\mathcal{K}_{\bmh}(\bmJ)=\left(\begin{array}{c}
\mathring{\mathcal{D}}^*(\bmJ)_{ij}\\
\epsilon^{ijk}D_iJ_{jkl}
\end{array}\right). \]
As remarked previously, a solution $(\bmK,\bar{\bmS},\bmS,\bmh)$ furnished in Step (i) gives rise to a zero quantity $J_{ijk}$ satisfying equations  \eqref{HavingSolvedExtrinsicCurvAuxiliaryEq} and \eqref{IntegrabilityConditionExtrCurv}, and which therefore lies in the kernel of the operator $\mathcal{K}_{\bmh}$ ---that is to say, $\mathcal{K}_{\bmh}(\bmJ)=0$. In order to establish that $J_{ijk}=0$ (see point (b), above), it suffices to show that $\mathcal{K}_{\bmh}$ has a trivial kernel. To do so, we aim to first establish injectivity of the operator $\mathcal{K}_{\mathring{\bmh}}$, and then to show that injectivity is preserved provided the metric $\bmh$ is sufficiently close to $\mathring{\bmh}$, in the appropriate norm. This ``stability" property of the kernel of $\mathcal{K}_{\bmh}$ relies crucially on the observation that the operator is, in fact, first-order elliptic ---see Lemma \ref{Lemma:EllipticityOfCurlyKOperator} and Proposition \ref{Prop:StabilityCodazziOperator} in Sections \ref{Subsec:PreliminaryRemarksSufficiency} and \ref{Subsubsection:SufficiencyMainArgument}.   
\\

On the other hand, note that
\begin{align*}
D^{i}(\mathcal{L}_Q \bmh)_{ij}-\tfrac{1}{2}D_j(\mathcal{L}_Q \bmh)_i{}^i &=D^i\left(D_iQ_j+D_jQ_i-D^kQ_k h_{ij}\right)\\
&=\Delta_{\bmh} Q_j+D^iD_jQ_i-D_jD^kQ_k\\
&=\Delta_{\bmh} Q_j+\left(D_jD^iQ_i+r_{ij}Q^i\right)-D_jD^kQ_k\\
&=\Delta_{\bmh} Q_j+r_{ij}Q^i.
\end{align*}
Therefore, if $J_{ijk}=0$, then \eqref{IntegrabilityConditionMetric}
implies the elliptic equation 
\[ \Delta_{\bmh} Q_j+r_{ij}Q^i=0,\]
for the zero quantity $Q_i$. Integrating by parts over the closed manifold $\mathcal{S}$, it follows that 
\begin{equation}
%-\int_{\mathcal{S}}Q^jD^{i}\mathcal{G}(\mathcal{L}_Q g)_{ij}~d\mu_\bmh
\int_{\mathcal{S}}\left(\Vert D\bmQ\Vert_\bmh^2-r_{ij}Q^iQ^j\right) ~d\mu_\bmh=0. \label{SufficiencyIntegralForQ}
\end{equation}
Note that the above identity only follows once it has been established
that $J_{ijk}=0$. Fortunately, the equation $\mathcal{K}_{\bmh}(\bmJ)=0$ is decoupled from $Q_i$ as a consequence of the semi-decoupling of
\eqref{IntegrabilityConditionMetric}--\eqref{IntegrabilityConditionExtrCurv},
as described in Remark \ref{Remark:EMVersionIsCoupled}. This decoupling allows for
a two step approach in which we first show $J_{ijk}=0$ and then use
\eqref{SufficiencyIntegralForQ} to show $Q_i=0$. The full argument is given in Proposition 
\ref{PropositionSufficiency} of Section \ref{Subsubsection:SufficiencyMainArgument}.

\subsection{Obstructions to the existence of solutions}
\label{Section:ObstructionsToSolutions}

In order to use the implicit function theorem (see Section
\ref{Section:ExistenceOfSolutions}) to establish existence of
solutions to the auxiliary system
\[
\Psi=0,
\]
 one would like to show that the linearisation $D\Psi$ is an
 isomorphism between suitable Banach spaces. Accordingly, by an
 \emph{obstruction to the existence of solutions}, we mean a
 non-trivial element of either $\text{ker}(D\Psi)$ or
 $\text{coker}(D\Psi)$ ---recalling that $D\Psi$ is an elliptic (and
 hence Fredholm) operator, the existence of a non-trivial cokernel is
 precisely the obstruction to surjectivity of $D\Psi$ while the
 existence of a non-trivial kernel is the obstruction to injectivity. 

As it will be seen, among the potential obstructions to the existence
of solutions one has non-trivial conformal Killing vectors and
tracefree Codazzi tensors of the background manifold. Precluding the
existence of such obstructions is the fundamental motivation behind
our choice of background data.
% ---as it will be seen in the subsequent section.

\begin{remark}
\em{It is not clear whether the obstructions that will be identified
in the sequel are essential, or may be circumvented. In
\cite{But06,But07}, for instance, the method follows through despite
the existence of non-trivial conformal Killing vectors. There, in
\emph{Step $(i)$} the auxiliary system is solved only up to an
\emph{error term}, constrained to lie in a finite-dimensional
space. In \emph{Step $(ii)$}, it is then simultaneously shown that the
error term must necessarily vanish and that the extended constraints
are indeed satisfied, as a consequence of the non-linear integrability
conditions
\eqref{IntegrabilityConditionExtrCurv}-\eqref{IntegrabilityConditionMetric}. Whether
such a procedure may be implemented in general is unclear. One might
expect the method to be more rigid in the compact case ---the non-existence of conformal 
Killing vectors, for instance, may be a prerequisite. An analogy may
be drawn here with the problem of \emph{linearisation stability} of
the constraint equations, in which the obstructions to
  integrability are precisely the so-called \textit{KID sets}, describing the projection onto $\mathcal{S}$
  of a spacetime Killing vector. In the case of non-compact $\mathcal{S}$, a solution of
the constraint equations may still be linearisation stable even when it admits a KID set, at least when the perturbations of the initial data are restricted to those of sufficiently fast decay 
at infinity (see for example \cite{Bar05}), while
the compact case is more rigid.}
\end{remark}

%%%%%%%%%%%%%%%%%%%%%%

\subsubsection{Conformal Killing vectors}
\label{Subsection:CKVectorsAsObstructions}
It is clear from the construction of the auxiliary system that the
existence of a non-trivial conformal Killing vector in the background
Riemannian manifold $(\mathcal{S},\mathring{\bmh})$, $\eta_i$ say,
destroys the injectivity of $D\Psi$, because of the use of the ansatz
\eqref{ModifiedYorkAnsatzElectric}-\eqref{ModifiedYorkAnsatzMagnetic}. Indeed,
$\text{ker}(D\Psi)$
contains linear combinations of 
\[
  (\sigma_{ij},~\bar{\xi}_i,~\xi_i,~\gamma_{ij})=(\bm0,~\eta_i,~\bm0,~\bm0)
  \qquad \text{and}\qquad
  (\sigma_{ij},~\bar{\xi}_i,~\xi_i,~\gamma_{ij})=(\bm0,~\bm0,~\eta_i,~\bm0)
  .
\] 
Moreover, in the case of a constant mean curvature background, the
second component of $D\Psi$ takes the form
\[
\mathring{\delta}(\mathring{L}(\bar{\bm\xi}))=0
\]
 and therefore in this
case $\text{coker}(D\Psi)$ also contains elements of the form
\[ 
(\sigma_{ij},~\bar{\xi}_i,~\xi_i,~\gamma_{ij})^*=(\bm0,~\eta_i,~\bm0,~\bm0),
\]
so that $D\Psi$ also fails to be surjective ---here we are using the
suffix ${}^*$ as a shorthand to denote an arbitrary element of the
codomain of $D\Psi$. Similar difficulties arise in both the conformal
method and the gluing methods, whenever there exist non-trivial
conformal Killing vectors ---see, for instance,
\cite{BarIse04}.

\begin{remark}
{\em From the previous discussion, it follows that the implementation of
the Friedrich--Butscher method will be simplified if one restricts to \emph{background initial
data sets which do not admit a conformal Killing vector}. This condition
holds, in particular, for manifolds of negative-definite Ricci
curvature ---the conformal Killing equation implies after contraction
with $D^i\eta^j$ and integration by parts that
\[
\int_{\mathcal{S}}\left(\Vert \mathring{D}\bmeta\Vert_{\mathring{\bmh}}^2+\tfrac{1}{3}\vert\mathring{\delta}(
  \bmeta)\vert^2-\mathring{r}_{ij}\eta^i\eta^j\right)~ d\mu_{\mathring{\bmh}} =0.
\]
Thus,  if the Ricci tensor is negative-definite then $\eta_i=0$ as a
consequence of the positive-definiteness
of the integrand. This is valid in particular for Einstein metrics of 
negative scalar curvature, despite them being locally maximally-symmetric ---that 
is to say that, while there exists the maximal number of \emph{local} Killing 
vector fields in a neighbourhood of each point, none may be extended globally to
the whole manifold. A sufficient condition for the stronger requirement of 
non-existence of local conformal Killing vector fields is given in \cite{BeiChrSch05}.
}
\end{remark}

\subsubsection{Non-trivial tracefree Codazzi tensors}
Inspection of the auxiliary equation for the
extrinsic curvature, equation \eqref{ExtrinsicCurvatureAuxiliary},
readily shows that the existence of non-trivial tracefree
\textit{Codazzi} tensors in the background initial data set ---i.e. elements of
$\text{ker}(\mathring{\mathcal{D}})\cap\mathscr{S}^2_0(\mathcal{S};\bmh)$---
also give rise to obstructions similar in nature to those arising from
the existence of conformal Killing vectors. In this case, given a
tracefree Codazzi tensor, $\eta_{ij}$ say, $\text{ker}(D\Psi)$ and
$\text{coker}(D\Psi)$ both contain elements of the form
\[
(\eta_{ij},~\bm0,~\bm0,~\bm0)
\]
 which destroy both the injectivity and the
surjectivity of $D\Psi$. 

\medskip
For examples of initial data sets which \emph{do} admit tracefree
Codazzi tensors, one needs only consider umbilical, conformally-flat
initial data sets. Consider
$(\mathcal{S},\mathring{\bmh},\mathring{\bmK}=\tfrac{1}{3}\mathring{K}\mathring{\bmh})$,
$\mathring{K}$ a constant, which constitutes an \textit{umbilical}
initial data set provided
\[
\mathring{r}=2\lambda-\tfrac{2}{3}\mathring{K}^2.
\]
If we restrict to those metrics $\mathring{\bmh}$ which are, in
addition, conformally flat then it follows from the Weyl-Schouten
Theorem (see Theorem 5.1 in \cite{CFEBook}) that
\[
0=\mathcal{H}_{ij}\equiv \mathring{\epsilon}_{kl(i}\mathring{D}^k \mathring{r}_{j)}{}^l\equiv
\mathring{\epsilon}_{kl(i}\mathring{D}^k \mathring{d}_{j)}{}^l,
\]
where $\mathring{d}_{ij}$ denotes the tracefree part of the Ricci
curvature. Moreover, it follows from the contracted second Bianchi
identity that $\delta_{\mathring{\bmh}}(\mathring{\bmd})_i=0$, again using the fact
that $\mathring{r}$ is constant. Combining the above observations it
follows (see Remark \ref{Remark:JacobiDecompOfCodazziOperator}) that
$\mathring{d}_{ij}$ is a tracefree Codazzi tensor
---i.e. $\mathring{\mathcal{D}}(\mathring{\bmd})_{ijk}=0$. \emph{This
Codazzi tensor is non-trivial (i.e. non-zero) if $\mathring{\bmh}$ is not an Einstein
metric.}

\begin{remark}
{\em The above observation is pertinent also to the case of non-compact
$\mathcal{S}$. In particular, it suggests that the time-symmetric
initial data set for the Schwarzschild spacetime, with metric
\[
\mathring{\bmh}=\left(1+\frac{m}{2r}\right)^4\bm\delta, 
\] is
potentially unsuitable (as background initial data) for the
application of the Friedrich--Butscher method as $\mathring{\bmh}$ is
not an Einstein metric.}
\end{remark}

%
%We will see in Section \ref{Subsubsection:SufficiencyMainArgument} ---see Proposition \ref{Prop:StabilityCodazziOperator}--- that non-existence of tracefree Codazzi tensors is, 
%in a sense, stable under perturbations of the metric. This will be used in the sufficiency 
%argument of the same section ---see Proposition \ref{PropositionSufficiency}.

\subsubsection{Conformally rigid hyperbolic manifolds}

From the previous two sections, we know that the existence of either a
non-trivial conformal Killing vector or a non-trivial tracefree
Codazzi tensor is undesirable for the application of the Friedrich--Butscher method
on compact manifolds. Moreover, it was noted in Section
\ref{Subsection:CKVectorsAsObstructions} that a Riemannian manifold of
negative-definite Ricci curvature cannot admit a globally-defined
conformal Killing field, rendering such a manifold a natural first
candidate for the background manifold $(\mathcal{S},\mathring{\bmh})$.

\medskip
Due to the highly-coupled nature of the auxiliary system of equations,
$\Psi=0$, the tractability of the required analysis is, of course,
dependent on the specific properties of the background manifold,
$(\mathcal{S},\mathring{\bmh})$. In particular, if we consider a
manifold $(\mathcal{S},\mathring{\bmh})$ that is Einstein (or,
equivalently, a \emph{space form} since we are in dimension $3$):
\[\mathring{r}_{ij}=\tfrac{1}{3}\mathring{r}\mathring{h}_{ij},\] with
$\mathring{r}$ (necessarily) constant,
%, or, equivalently, of constant sectional curvature, $k$, 
then $D\Psi$ simplifies significantly. The
requirement that $\mathring{r}_{ij}$ be negative-definite is then
simply that $\mathring{r}$ be negative.

% \begin{remark}
% {\em Note that, although an Einstein manifold in
% dimension $3$ is necessarily (locally) conformally flat, the tracefree
% Ricci curvature, being zero, is not an example of a non-trivial
% tracefree Codazzi tensor ---see the discussion in the previous
% section.}
% \end{remark}

Accordingly, let us restrict to an Einstein background manifold with negative 
Ricci scalar
---we will refer to such a manifold as \emph{hyperbolic}. Recall that,
by the \emph{Killing--Hopf Theorem} $(\mathcal{S},\mathring{\bmh})$ is
isometric to a quotient of the hyperbolic $3$-space $\mathbb{H}^3$.
We refer the reader to \cite{Bes08} for results concerning the
admissible topologies of $\mathcal{S}$.
%So far we have guaranteed the non-existence of a non-trivial conformal Killing vector, but we would also like to eliminate the possibility of 
 Moreover, we would also like to exclude the possibility of a
non-trivial tracefree Codazzi tensor ---i.e. ensure that
$\text{ker}(\mathring{\mathcal{D}})\cap
\mathscr{S}^2_0(\mathcal{S};\mathring{\bmh})=\lbrace 0\rbrace$. Now,
in the case of hyperbolic manifolds ---see \cite{Laf83} and also also
\cite{Bei97}--- the space of tracefree Codazzi tensors coincides with
the space of \emph{essential conformally flat deformations}
---i.e. one has 
\[
\text{ker}\lbrace\mathring{\mathcal{D}}:
\mathscr{S}_0^2(\mathcal{S};\mathring{\bmh})\rightarrow\mathcal{J}(\mathcal{S})\rbrace=\text{ker}~\mathring{H}\cap\text{ker}~\mathring{\delta}\simeq\text{ker}~\mathring{H}/\mathring{L}(\Lambda^1(\mathcal{S})),
\]
where $\mathring{H}$ denotes the \emph{linearised Cotton map} ---see
Section \ref{Section:ParametrisingFreeData} for more
details. Consequently, we will refer to a hyperbolic manifold which
admits no no-trivial tracefree Codazzi tensors as being
\emph{conformally rigid}. The requirement of conformal rigidity places
additional restrictions on the topology of $\mathcal{S}$, but there
remains a non-empty family of such manifolds ---see \cite{Kap94}.

\section{Nonlinear perturbations of compact hyperbolic initial data}
\label{Section:PerturbationsOfCompactHyperbolicInitialData}

In the remainder of this article we restrict our attention to conformally rigid hyperbolic background initial data, since such manifolds admit neither conformal Killing fields nor tracefree Codazzi tensors.

The results here can be thought of spatially-closed analogues of those in \cite{Delay09}, in which a version of the Friedrich--Butscher method was applied to non-compact hyperbolic background manifolds. 
We note however that here we solve the full extended constraint equations, rather than the reduced 
system corresponding to initial data sets of umbilic extrinsic curvature, as considered in \cite{Delay09} ---i.e. we allow for non-trivial perturbations of the extrinsic curvature. 

\subsection{Statement of the main result}
\label{Sec:MainResult}
In the following, let $(\mathcal{S},\mathring{\bmh})$ be a closed
hyperbolic Einstein manifold with sectional curvature normalised to
$k=-1$ (or, equivalently, with $\mathring{r}=-6$). Then, for any
given constant $\mathring{K}$, the tensor fields 
\begin{equation}
\mathring{h}_{ij},\qquad \mathring{K}_{ij}=\tfrac{1}{3}\mathring{K} \mathring{h}_{ij},\label{HyperbolicInitialData}
\end{equation}
over $\mathcal{S}$ constitute a solution to the Einstein constraint
equations with constant mean extrinsic curvature  $\mathring{K}$ and
with cosmological constant given by
\[
\lambda=\tfrac{1}{3}(\mathring{K}^2-9),
\]
as it can be readily seen from the Hamiltonian constraint
\eqref{HamiltonianConstraint}. Initial data of this type will be
called \emph{hyperbolic initial data}.  The Cauchy stability of the
development of initial data sets of this type, with $\lambda=0$, was
studied in \cite{AndMon04}.

\begin{remark}
{\em Note that here we are choosing to normalise the intrinsic
curvature, which in turn fixes the value of the cosmological constant,
once the extrinsic curvature has been given. One could alternatively
rescale the intrinsic and extrinsic curvatures appropriately so as to
normalise the cosmological constant. The former option is chosen
since, in the subsequent analysis, it is the intrinsic geometry of
$(\mathcal{S},\mathring{\bmh})$ that will be of primary importance.}
\end{remark}

\begin{remark}{\em
The (unique) solution to the extended Einstein
constraint equations associated to \eqref{HyperbolicInitialData}
is obtained by setting $\mathring{S}_{ij}=\mathring{\bar{S}}_{ij}=0$ ---see \eqref{Eq:ElectricPartFromInitialData}--\eqref{Eq:MagneticPartFormInitialData}. Note that the sign of
$\lambda$ is dependent on the choice of $\mathring{K}$: $\lambda<0$
for $\vert\mathring{K}\vert<3$, $\lambda=0$ for $\mathring{K}=\pm 3$
and $\lambda>0$ for $\vert\mathring{K}\vert>3$.}
\end{remark}

In the following it will prove convenient to define the constants
\begin{equation}
\alpha\equiv
-4+\frac{2}{9}\mathring{K}^2,\qquad\beta\equiv 
-4+\frac{8}{9}\mathring{K}^2.
\label{Definition:AlphaBeta}
\end{equation}
Define also for $s\geq 4$ the Banach spaces
%\footnote{The fact that $H^l(\mathscr{S}_{TT}(\mathcal{S},\mathring{\bmh}))$ is indeed a sub-Banach space of $H^l(\mathscr{S}^2_{0}(\mathcal{S};\mathring{\bmh}))$ follows fro by the ``Splitting Lemma", Lemma \ref{Lemma:Splitting}.} 
$\mathcal{X}^s,\mathcal{Y}^s,\mathcal{Z}^s$, as follows
\begin{eqnarray*}
&& \mathcal{X}^s\equiv H^{s-1}(\mathscr{C}(\mathcal{S}))\times H^{s-1}(\mathscr{S}_{TT}(\mathcal{S};\mathring{\bmh}))\times H^{s-1}(\mathscr{S}_{TT}(\mathcal{S};\mathring{\bmh})),\\
&& \mathcal{Y}^s\equiv H^s(\mathscr{S}^2_0(\mathcal{S};\mathring{\bmh}))\times H^s(\Lambda^1(\mathcal{S}))\times H^s(\Lambda^1(\mathcal{S}))\times H^s(\mathscr{S}^2(\mathcal{S})), \\
&& \mathcal{Z}^s\equiv H^{s-2}(\mathscr{S}^2_0(\mathcal{S};\mathring{\bmh}))\times H^{s-2}(\Lambda^1(\mathcal{S}))\times H^{s-2}(\Lambda^1(\mathcal{S}))\times H^{s-2}(\mathscr{S}^2(\mathcal{S})). 
\end{eqnarray*}
and where the norms are defined with respect to the background metric
$\mathring{\bmh}$ ---unless explicitly indicated otherwise, all
$H^s$-norms from now on will be defined with respect to
$\mathring{\bmh}$.
%\mnotex{Juan: Could you also check this, please?}
\begin{remark}
{\em 
That the image of $\Psi: \mathcal{X}^s\times\mathcal{Y}^s$ is indeed contained in $\mathcal{Z}^s$ may be easily checked using the Schauder ring property: namely that $(\bmu,\bmv)\mapsto \bmu\otimes\bmv$ is continuous as a mapping from $H^{s_1}\times H^{s_2}$ to $H^{s_3}$ provided $s_1+s_2>s_3+n/2$ and $s_1,s_2>s_3$ ---see \cite{Cho08}, for instance. }
\end{remark}
%Define also the Banach spaces $\mathcal{X},\mathcal{Y}$ as follows
%\[
%\mathcal{X}\equiv\mathcal{B}_\phi\times\mathcal{B}_\bmT\times\mathcal{B}_{\bar{\bmT}},\qquad
%\mathcal{Y}\equiv \mathcal{B}_{\bmchi}\times\mathcal{B}_\bmX\times\mathcal{B}_{\bar{\bmX}}\times\mathcal{B}_{\bmh},
%\]
%defined in terms of the Sobolev spaces 
%\[
%\mathcal{B}_\bmh = H^4(\mathscr{S}^2(\mathcal{S})), \qquad \mathcal{B}_\bmchi=  L^2(\mathscr{S}^2_0(\mathcal{S};\mathring{\bmh})), \qquad \mathcal{B}_\bmX= \mathcal{B}_{\bar{\bmX}}= L^2(\Lambda^1(\mathcal{S})), \]
%\[ \mathcal{B}_\phi= L^2(\mathcal{S}),\qquad
% \mathcal{B}_\bmT= \mathcal{B}_{\bar{\bmT}}= L^2(\mathscr{S}_{TT}(\mathcal{S};\mathring{\bmh})),\]
%\begin{align*}
%& \mathcal{B}_\bmh = H^4(\mathscr{S}^2(\mathcal{S})), \\
%& \mathcal{B}_\bmchi=  L^2(\mathscr{S}^2_0(\mathcal{S};\mathring{\bmh})),\\
%& \mathcal{B}_\bmX= \mathcal{B}_{\bar{\bmX}}= L^2(\Lambda^1(\mathcal{S})),\\
%& \mathcal{B}_\phi= L^2(\mathcal{S}),\\
%& \mathcal{B}_\bmT= \mathcal{B}_{\bar{\bmT}}= L^2(\mathscr{S}_{TT}(\mathcal{S};\mathring{\bmh})),
%\end{align*}
% \mnotex{I think the metric, in particular, need to be more regular for
%   the sufficiency argument to follow through?}
%and where the norms are defined with respect to the background metric
%$\mathring{\bmh}$ ---unless explicitly indicated otherwise, all
%$H^s$-norms from now on will be defined with respect to
%$\mathring{\bmh}$.
%Note that the Banach space associated to the
%metric, $\mathcal{B}_\bmh$, is split into two parts: the trace and
%tracefree parts with respect to $\mathring{\bmh}$.

\medskip
We are now in a position to state our main theorem:

\begin{theorem}
\label{MainTheorem}
Let $(\mathcal{S},\mathring{\bmh},\mathring{\bmK})$ be a smooth
conformally rigid hyperbolic initial data set with constant mean
extrinsic curvature $\mathring{K}$ satisfying
\begin{equation}
\beta\notin\text{\em
  Spec}\big(-\mathring{\Delta}:C^\infty(\mathcal{S})\rightarrow
C^\infty(\mathcal{S})\big).  \label{ConditionsOnMeanCurvature}
\end{equation} 
%\begin{equation}
%\alpha_{\mathring{K}}\in\mathbb{R}\setminus\text{spec}(\mathring{\Delta}_L)\label{ConditionsOnMeanCurvature}
%\end{equation}  
Then, there exists an open neighbourhood $\mathcal{U}\subset
\mathcal{X}$ of $(\bmzero,\bmzero,\bmzero)$, an open neighbourhood
$\mathcal{W}\subset\mathcal{Y}$ of
$(\mathring{\bmh},\bmzero,\bmzero,\mathring{\bmK})$ and a smooth map
$\nu: \mathcal{U}\rightarrow \mathcal{W}$
%\[\begin{array}{cccc}
%\nu: &\mathcal{U}&\longrightarrow & \mathcal{W}\\
%& u:=(\phi,T_{ij},\bar{T}_{ij}) & \mapsto & (h(u),\chi(u),X(u),\bar{X}(u))
%\end{array} \]
such that, defining 
\[
u\equiv (\phi,\bmT,\bar{\bmT}),\qquad \nu(u)\equiv
\big(\bmchi(u),\bar{\bmX}(u),\bmX(u),\bmh(u)\big),
\]
the following assertions hold:

\begin{enumerate}[i)]

\item for each $(\phi,\bmT,\bar{\bmT})\in\mathcal{U}$,
\[
w(u)\equiv\big(\bm\chi(u)+\tfrac{1}{3}(\phi+\mathring{K})
  \mathring{\bmh},~\bar{\bmS}(\bar{\bmX}(u),\bar{\bmT}),~
  \bmS(\bmX(u),\bmT),~\bmh(u)\big)
\]
is a solution to the extended constraint equations \eqref{ExtendedConstraints} with cosmological
constant $\lambda=(\mathring{K}^2-9)/3$;

\item the map $u\mapsto w(u)$ is injective for
$\mathring{K}\neq 0$. Moreover, it is injective for $\mathring{K}=0$
if we restrict the free datum $\phi$ to the sub-Banach space
of functions which integrate to zero over $\mathcal{S}$ ---that is to say that each
such solution $w$ corresponds to a unique choice of free data
$u=(\phi,\bmT,\bar{\bmT})$.
%\item each choice of the free data $T_{ij},~\bar{T}_{ij}$ may be expressed uniquely (up to the addition of Killing forms) in the form
%\[T_{ij}=H(\eta)_{ij},\qquad \bar{T}_{ij}=H(\bar{\eta})_{ij},\]
%for some $\eta_{ij},~\bar{\eta}_{ij}$ symmetric and tracefree with respect to $\bmh$, where $H$ denotes the linearised Cotton map. 
\end{enumerate}
%\[\left(\begin{array}{cccc}
%\text{proj}_g\circ\mathcal{Y}  & 0 & 0 & 0\\
%0 & \text{Id} & 0 & 0
%\end{array}\right) \]
\end{theorem}

\begin{remark}
\em{Notice that when $\vert\mathring{K}\vert\leq \sqrt{9/2}$ ---and,
in particular in the time-symmetric case, $\mathring{K}=0$---
condition \eqref{ConditionsOnMeanCurvature} is satisfied trivially
since $\beta<0$ but $-\mathring{\Delta}$ is
positive-semi-definite. Note that in this case the cosmological
constant is negative ($\lambda<0$). Moreover, since the spectrum of
$-\mathring{\Delta}$ is discrete, condition
\eqref{ConditionsOnMeanCurvature} excludes only countably many values of $\mathring{K}$.}
\end{remark}

%\begin{remark}
%{\em Smoothness of the map $\nu$ is guaranteed by the implicit function theorem, as a consequence of the smoothness of $D\Psi$.}
%\end{remark}
% \mnotex{Is $\Psi$ really smooth? This is what Butscher claims (and is key to his argument)! Do we need this anymore? I've commented out this stuff.}
%\begin{remark}
%\em{Notice that when $\vert\mathring{K}\vert\leq \sqrt{9/2}$ ---and, in particular, in the time-symmetric case, $\mathring{K}=0$--- condition \eqref{ConditionsOnMeanCurvature} is satisfied trivially, since  $\alpha_{\mathring{K}}<\beta_{\mathring{K}}<0$, but $-\mathring{\Delta}$ and $\mathring{\Delta}_L$ are both positive-semi-definite. Note that in this case the cosmological constant is negative ($\lambda<0$). Moreover, since the spectra of $-\mathring{\Delta},~\mathring{\Delta}_L$ are discrete, condition \eqref{ConditionsOnMeanCurvature} excludes only countably-many values of $\mathring{K}$. Note that condition \eqref{ConditionsOnMeanCurvature} is equivalent to 
%\[ \notin\text{spec}(-\mathring{\Delta}:\mathscr{S}^2_0(\mathcal{S};\mathring{h})\rightarrow \mathscr{S}^2_0(\mathcal{S};\mathring{h})).\]
% }
%\end{remark} 

The theorem will be proven in two stages in the forthcoming sections,
by means of Propositions
\ref{PropExistenceOfSolutionsToSecondarySystem} and
\ref{PropositionSufficiency}. In Section
\ref{Section:ParametrisingFreeData} we describe a parametrisation of
the free data through the use of the linearised Cotton map, based on
the results of \cite{Bei97,GasGol84}, and summarised in Proposition
\ref{Prop:ParametrisingFreeData}.

\subsection{Existence of solutions of the auxiliary system}
\label{Section:ExistenceOfSolutions}

The purpose of this section is to show the existence of perturbative
solutions to the auxiliary system in the case of conformally rigid
hyperbolic initial data sets. 

\subsubsection{Technical tools}
The main tool used in establishing existence is the \emph{Implicit Function
Theorem} ---see e.g. \cite{Edw12}--- which we state
here for completeness.

\begin{theorem*}[\textbf{\em Implicit Function Theorem}] Let
  $\mathcal{X},~\mathcal{Y},~\mathcal{Z}$ be Banach spaces, and 
\[
\Psi: \mathcal{X}\times \mathcal{Y}\rightarrow \mathcal{Z}
\]
 a mapping with continuous Fr\'{e}chet derivative. Suppose that
$(x_0,y_0)\in \mathcal{X}\times \mathcal{Y}$ satisfies
$\Psi(x_0,y_0)=0$ and that the map $y\mapsto D\Psi(x_0,y_0)(0,y)$ is a
Banach space isomorphism from $\mathcal{Y}$ onto $\mathcal{Z}$. Then,
there exist open neighbourhoods $\mathcal{U}$ of $x_0$ and
$\mathcal{V}$ of $y_0$ and a Fr\'{e}chet-differentiable mapping $\nu:
\mathcal{U}\rightarrow \mathcal{V}$ such that $\Psi(x,\nu(x))=0$ for
all $x\in \mathcal{U}$, and $\Psi(x,y)=0$ for $(x,y)\in
\mathcal{U}\times \mathcal{V}$ if and only if $y=\nu(x)$. Moreover, if
the map $x\mapsto D\Psi(x_0,y_0)(x,0)$ is injective, then $\nu$ is
also injective. 
%If the map $\Psi$ \mnotex{Should this be $D\Psi$?}is moreover smooth, then the resulting
%map $\nu$ is also smooth.
\end{theorem*}
%\mnotex{Commented out smoothness stuff}
In order to establish that the various mappings of interest are
isomorphisms, we will make use of the following \emph{Splitting Lemma}
---see e.g. \cite{LawMic89}.  

\begin{lemma*}[\textbf{\em Splitting Lemma}]
Let $E$ and $F$ be vector bundles over $\mathcal{S}$, with fixed
Riemannian metric $\bmh$. Let
\[\mathscr{D}:C^\infty(E)\longrightarrow C^\infty(F) \] be a
differential operator of order $k$, and $\mathscr{D}^*$ the
corresponding formal $L^2$-adjoint. Suppose that $\mathscr{D}$ is
overdetermined elliptic (equivalently, $\mathscr{D}^*$ is
underdetermined elliptic), then for $s\in[k,\infty)$
\[
H^s(\mathcal{S})=\text{\em Im}~\mathscr{D^*}\oplus \text{\em ker}~\mathscr{D} ,
\]
where both factors are closed and are $L^2$-orthogonal. Moreover, if
$\mathscr{D}$ is injective, then $\mathscr{D}^*$ is surjective, and
the composition $\mathscr{D}^*\circ\mathscr{D}$ is an isomorphism.
\end{lemma*}

\subsubsection{The application of the Implicit Function Theorem}

Since the background solution admits no conformal Killing vectors and
no non-trivial tracefree Codazzi tensors, the operators $\mathring{L}$ and
$\mathring{\mathcal{D}}$ are both injective. Therefore, by the
Splitting Lemma, the following are isomorphisms for $s\geq 4$:
\begin{eqnarray*}
&& \mathring{\delta}\circ
   \mathring{L}:H^s(\Lambda^1(\mathcal{S}))\rightarrow
   H^{s-2}(\Lambda^1(\mathcal{S})),  \\
&& \mathring{\mathcal{D}}^*\circ\mathring{\mathcal{D}}:H^s(\mathscr{S}^2_0(\mathcal{S};\mathring{\bmh}))\rightarrow H^{s-2}(\mathscr{S}^2_0(\mathcal{S};\mathring{\bmh})).
\end{eqnarray*}
%\begin{eqnarray*}
%&& \mathring{\delta}\circ
%   \mathring{L}:\Lambda^1(\mathcal{S})\rightarrow
%   \Lambda^1(\mathcal{S}),  \\
%&& \mathring{\mathcal{D}}^*\circ\mathring{\mathcal{D}}:\mathscr{S}_0^2(\mathcal{S};\mathring{\bmh})\rightarrow \mathscr{S}_0^2(\mathcal{S};\mathring{\bmh}).
%\end{eqnarray*}
Since the background initial data, being hyperbolic, consists of an Einstein metric and umbilical extrinsic curvature, the linearisation of the auxiliary extended constraint map in the direction of the determined fields, $D\Psi$, takes the form
\[
D\Psi\cdot(\bmsigma,\bar{\bmxi},\bmxi,\bmgamma;\phi,\bar{\bmT},\bmT)=\left(\begin{array}{l}
\mathring{\mathcal{D}}^*\big(\mathring{\mathcal{D}}(\bmsigma)-\tfrac{1}{3}\mathring{K}\mathring{\mathcal{D}}(\bmgamma)-\mathring{\star}\mathring{L}(\bar{\bmxi})\big)_{ij}\\
\mathring{\delta} \circ \mathring{L}(\bar{\bmxi})_i\\
\mathring{\delta}\circ \mathring{L}(\bmxi)_i\\
\tfrac{1}{2}\mathring{\Delta}_L\gamma_{ij} - \tfrac{1}{2}\alpha \bar{\gamma}_{ij}  -  \tfrac{1}{6}\beta \gamma \mathring{h}_{ij} + \tfrac{1}{3}\mathring{K} \sigma_{ij}- \mathring{L}(\bmxi)_{ij}
\end{array}\right). 
\] 
%In particular, the linearisation around a background hyperbolic initial data set, as considered here, takes the above form.

\begin{remark}
{\em 
Let $(A_{ij},~\bar{B}_{i},~B_{i},~C_{ij})\in\mathcal{Z}^s$ be arbitrary. Then in
order to establish whether $D\Psi$ is an isomorphism, we are
concerned with solving the system of equations
\begin{subequations}
\begin{eqnarray}
&& \mathring{\mathcal{D}}^*(\mathring{\mathcal{D}}(\bmsigma)-\tfrac{1}{3}\mathring{K}\mathring{\mathcal{D}}(\bmgamma)-\mathring{\star}\mathring{L}(\bar{\bmxi}))_{ij}=A_{ij},\label{LinearisedAuxiliaryEqExtrinsicCurv}\\
&& \mathring{\delta} \circ \mathring{L}(\bar{\bmxi})_i=\bar{B}_i,\label{LinearisedEqMagnetic}\\
&& \mathring{\delta}\circ \mathring{L}(\bmxi)_i=B_i,\label{LinearisedEqElectric}\\
&& \mathring{\Delta}_L\gamma_{ij} - \alpha \bar{\gamma}_{ij}  -  \tfrac{1}{3}\beta \gamma \mathring{h}_{ij} + \tfrac{2}{3}\mathring{K} \sigma_{ij}-2 \mathring{L}(\bmxi)_{ij} =C_{ij},\label{LinearisedEqMetric}
\end{eqnarray}
\end{subequations}
where here $\gamma$ and $\bar{\gamma}_{ij}$ denote the \emph{trace and
tracefree parts} of $\gamma_{ij}$ with respect to $\mathring{\bmh}$,
and  the constants $\alpha,~\beta$ are as defined
in \eqref{Definition:AlphaBeta}. Note the semi-decoupled form of the system: one can first solve
\eqref{LinearisedEqMagnetic}-\eqref{LinearisedEqElectric}, and then
proceed to solve \eqref{LinearisedAuxiliaryEqExtrinsicCurv} and
\eqref{LinearisedEqMetric}, in turn.}
\end{remark}
%\mnotex{$A,B,C$ are not allowed to be smooth, just in the appropriate $H^s$ Banach spaces.}

In order to address injectivity if the map $\nu$, we also need to
consider the linearisation of $\Psi$ in the direction of the free
data. For a general data set $(\mathcal{S},\mathring{\bmh},\mathring{\bmK})$ the linearisation is given by 
\begin{equation}
\label{LinearisationInDirectionOfFreeData}
\frac{\mbox{d}}{\mbox{d}\tau}\Psi(\bmchi,\bar{\bmX},\bmX,\bmh;~\mathring{K}+\tau\phi,\tau\bar{\bmT},\tau \bmT) \bigg\vert_{\tau=0}=\left(\begin{array}{c}
-\tfrac{1}{6}\mathring{L}(d\phi)_{jk}- \tfrac{1}{2} \mathring{\epsilon}_{kil} \mathring{D}^{l}\bar{T}_{j}{}^{i} -  \tfrac{1}{2} \mathring{\epsilon}_{jil} \mathring{D}^{l}\bar{T}_{k}{}^{i}\\[0.5em]
\mathring{\epsilon}_{ljk} \mathring{K}^{ij} T_{i}{}^{k} + \mathring{D}^{i}\bar{T}_{il} \\[0.5em]
-\mathring{\epsilon}_{ikl} \mathring{K}^{jk} \bar{T}_{j}{}^{l}+ \mathring{D}^{j}T_{ij} \\[0.5em]
- T_{ij} + \tfrac{1}{3}( \mathring{K}_{ij} +\mathring{K} \mathring{h}_{ij})\phi  
\end{array}\right).
\end{equation}

\begin{remark}
{\em
It is clear that if the above map is to be injective then we should at least
require $T_{ij},\bar{T}_{ij}$ to be tracefree with respect to
$\mathring{\bmh}$ ---it is easy to verify that pure trace $T_{ij}$ and
$\bar{T}_{ij}$ would be in the kernel. This further justifies the use
of the ansatz
\eqref{ModifiedYorkAnsatzElectric}-\eqref{ModifiedYorkAnsatzMagnetic}.
}
\end{remark}

The existence of solutions to the auxiliary system is established in
the following proposition.

\begin{proposition}[\textbf{\em existence of solutions to the
    auxiliary system}]
\label{PropExistenceOfSolutionsToSecondarySystem}
Let $(\mathcal{S},\mathring{\bmh},\mathring{\bmK})$ be a smooth conformally rigid
hyperbolic initial data set with (constant) mean extrinsic curvature
$\mathring{K}$ satisfying condition \eqref{ConditionsOnMeanCurvature}. Then
$D\Psi:\mathcal{Y}^s\rightarrow\mathcal{Z}^s$ is a Banach space
isomorphism for $s\geq 4$, and so (by the implicit function theorem) there exist open neighbourhoods
$(\mathring{K},\bm0,\bm0)\in\mathcal{V}\subset\mathcal{Y}^s$ and
$(\mathring{\bmK},\bm0,\bm0, \mathring{\bmh})\in\mathcal{U}\subset\mathcal{X}^s
$ and a Fr\'{e}chet differentiable map
$\nu:\mathcal{U}\rightarrow\mathcal{V}$ mapping free data to solutions
of the auxiliary system $\Psi=0$. Moreover the map $\nu$ is
injective.
% and takes smooth free data to smooth solutions of $\Psi=0$.
\end{proposition}
%\mnotex{Commented out smoothness stuff.}
%\mnotex{Added the smoothness comment at the end}

\begin{proof}
$\phantom{X}$

\smallskip
\noindent
\textbf{Injectivity of} $\mathbf{D\Psi}$. Taking
$A_{ij}=C_{ij}=0,~B_i=\bar{B}_i=0$ in equations
\eqref{LinearisedAuxiliaryEqExtrinsicCurv}-\eqref{LinearisedEqMetric}, we aim to show triviality of solutions $(\bm\sigma,\bar{\bm\xi},\bm\xi,\bmgamma)$. Note that by elliptic regularity (see Appendix I of \cite{Bes08}, for instance), it suffices to show restrict to smooth $(\bm\sigma,\bar{\bm\xi},\bm\xi,\bmgamma)$. Equations
\eqref{LinearisedEqMagnetic}-\eqref{LinearisedEqElectric} imply,
firstly, that $\xi_i=\bar{\xi}_i=0$ since the background
metric admits no global conformal Killing vectors. Substituting into
\eqref{LinearisedAuxiliaryEqExtrinsicCurv} and
\eqref{LinearisedEqMetric}
\begin{subequations}
\begin{eqnarray}
&& \mathring{\mathcal{D}}^*\circ{\mathring{\mathcal{D}}}(\bmsigma-\tfrac{1}{3}\mathring{K}\bmgamma)_{ij}=0, \label{LinearisedAuxiliaryEqExtrinsicCurvSimplified}\\
&& \mathring{\Delta}_L\gamma_{ij} - \alpha \bar{\gamma}_{ij}  -  \tfrac{1}{3}\beta \gamma \mathring{h}_{ij} + \tfrac{2}{3}\mathring{K} \sigma_{ij} =0.\label{LinearisedEqMetricSimplified1}
\end{eqnarray}
\end{subequations}
Tracing \eqref{LinearisedEqMetricSimplified1} we obtain  
\[ 
- (\mathring{\Delta}+\beta )\gamma =0.
\]
By assumption
$\beta\notin\text{Spec}(-\mathring{\Delta})$ and
therefore $\gamma=0$. Substituting into
\eqref{LinearisedAuxiliaryEqExtrinsicCurvSimplified}
\begin{equation}
\mathring{\mathcal{D}}^*\circ{\mathring{\mathcal{D}}}(\bmsigma-\tfrac{1}{3}\mathring{K}\bar{\bmgamma})_{ij}=0.
\end{equation}
Now, since
$\mathring{\mathcal{D}}^*\circ\mathring{\mathcal{D}}:\mathscr{S}^2_0(\mathcal{S};\mathring{\bmh})\rightarrow
\mathscr{S}^2_0(\mathcal{S};\mathring{\bmh})$ is an isomorphism,
$\sigma_{ij}=\tfrac{1}{3}\mathring{K}\bar{\gamma}_{ij}$. Substituting into
\eqref{LinearisedEqMetricSimplified1} along with $\gamma=0$ yields
\begin{equation}
\mathring{\Delta}_L\bar{\gamma}_{ij} +4\bar{\gamma}_{ij}\equiv -\mathring{\Delta}\bar{\gamma}_{ij}-2\bar{\gamma}_{ij} =0. \label{LinearisedEqMetricSimplified2}
\end{equation} 
% \mnotex{I think I've plugged a hole here ---was still using the
%   claimed positive-semi-definiteness of $\Delta_L$!}

We will now show that
$(\mathring{\Delta}_L+4):\mathscr{S}^2_0(\mathcal{S};\mathring{h})\rightarrow
\mathscr{S}^2_0(\mathcal{S};\mathring{h})$ is injective (and hence, by
self-adjointness, an isomorphism). First, taking the divergence of
\eqref{LinearisedEqMetricSimplified2}, commuting derivatives and using
the fact that the background metric is Einstein (with $\mathring{r}=-6$), we find that
\begin{align*}
0&=-\mathring{D}^i(\mathring{\Delta}\bar{\gamma}_{ij}+2\bar{\gamma}_{ij})\\
&=-\mathring{\Delta}\mathring{\delta}(\bar{\bm\gamma})_j-\mathring{D}^k(\mathring{r}_k{}^l\bar{\gamma}_{lj}-\mathring{r}_j{}^l{}_k{}^i\bar{\gamma}_{il})-\mathring{r}_j{}^{lik}\mathring{D}_k\bar{\gamma}_{il}-2\mathring{\delta}(\bar{\bm\gamma})_j\\
&=-\mathring{\Delta}\mathring{\delta}(\bar{\bm\gamma})_j-\mathring{r}^{kl}\mathring{D}_k\bar{\gamma}_{lj}-2\mathring{r}_j{}^{lik}\mathring{D}_k\bar{\gamma}_{il}-2\mathring{\delta}(\bar{\bm\gamma})_j\\
&=(-\mathring{\Delta}+2)\mathring{\delta}(\bar{\bm\gamma})_j,
\end{align*}
and hence we see that $\mathring{\delta}(\bar{\bm\gamma})=0$ by positivity
of
$(-\mathring{\Delta}+2):\Lambda^1(\mathcal{S})\rightarrow\Lambda^1(\mathcal{S})$. Now,
\begin{align*}
\mathring{\mathcal{D}}^*\circ\mathring{\mathcal{D}}(\bm\bar{\gamma})_{ij}&=\mathring{\Delta}\bar{\gamma}_{ij}-\tfrac{1}{2}\mathring{D}_k\mathring{D}_i\bar{\gamma}_j{}^k-\tfrac{1}{2}\mathring{D}_k\mathring{D}_j\bar{\gamma}_i{}^k+\tfrac{1}{3}\mathring{D}^k\mathring{D}^l\bar{\gamma}_{kl}\mathring{h}_{ij}\\
&=\mathring{\Delta}\bar{\gamma}_{ij}-\tfrac{1}{2}\mathring{D}_i\mathring{D}_k\bar{\gamma}_j{}^k-\tfrac{1}{2}\mathring{D}_j\mathring{D}_k\bar{\gamma}_i{}^k+\tfrac{1}{3}\mathring{D}^k\mathring{D}^l\bar{\gamma}_{kl}\mathring{h}_{ij}+3\bar{\gamma}_{ij}\\
&=-(\mathring{\Delta}_L+4)\bar{\gamma}_{ij}+\bar{\gamma}_{ij}\\
&=\bar{\gamma}_{ij},
\end{align*}
where in the third line we are using
$\mathring{\delta}(\bar{\bm\gamma})=0$ and in the fourth we are using
\eqref{LinearisedEqMetricSimplified2}. However, clearly
$\mathring{\mathcal{D}}^*\circ \mathring{\mathcal{D}}$ is
negative-definite, and so we find that $\bar{\gamma}_{ij}=0$ ---that
is to say, $(\mathring{\Delta}_L+4)$ is injective.
%from which we see that $\bar{\gamma}_{ij}=0$ by positive
%semi-definiteness of $\mathring{\Delta}_L$. 
Collecting everything
together, we have found that 
\[
\sigma_{ij}=\gamma_{ij}=0, \qquad  \xi_i=\bar{\xi}_i=0,  
\]
---i.e. the map $D\Psi$ is injective.

\smallskip
\noindent
\textbf{Surjectivity of} $\mathbf{D\Psi}$. The argument for
surjectivity is similar. First, since $\mathring{\delta}\circ
\mathring{L}$ is an isomorphism, equations
\eqref{LinearisedEqMagnetic}-\eqref{LinearisedEqElectric} admit
(unique) solutions $\bar{\xi}_i,~\xi_i$, for any given
$\bar{B}_{i},~B_{i}$. Substituting into equations 
\eqref{LinearisedAuxiliaryEqExtrinsicCurv} and \eqref{LinearisedEqMetric}
and rearranging one obtains 
\begin{subequations}
\begin{eqnarray}
&& \mathring{\mathcal{D}}^*\circ\mathring{\mathcal{D}}(\bmvarsigma-\tfrac{1}{9}\mathring{K}\gamma\mathring{\bmh})_{ij}=A_{ij}+\mathring{\mathcal{D}}^*(\mathring{\star} \mathring{L}(\bar{\bmxi})),\label{LinearisedEqExtrinsicCurv}\\
&& \mathring{\Delta}_L\gamma_{ij} +4 \bar{\gamma}_{ij}  -  \tfrac{1}{3}\beta \gamma \mathring{h}_{ij}+ \tfrac{2}{3}\mathring{K} \varsigma_{ij}=C_{ij}+2 \mathring{L}(\bmxi)_{ij}, \label{LinearisedEqMetricSimplified3} 
\end{eqnarray}
\end{subequations}
where, for simplicity, we have defined
\[
\varsigma_{ij}\equiv
\sigma_{ij}-\tfrac{1}{3}\mathring{K}\bar{\gamma}_{ij}.
\]
Note that $\varsigma_{ij}$ is tracefree with respect to $\mathring{\bmh}$.
%Since $\mathring{\mathcal{D}}^*\circ\mathring{\mathcal{D}}:\mathscr{S}^2_0(\mathcal{S};\mathring{h})\rightarrow \mathscr{S}^2_0(\mathcal{S};\mathring{h})$ is an isomorphism, there exists a unique solution $\sigma_{ij}$ for any given $A_{ij}$, and which is of course dependent on the $\bar{\xi}_i$ obtained in the previous step. Substituting into \eqref{LinearisedEqMetric},
%\begin{equation}
%\mathring{\Delta}_L\gamma_{ij} - \alpha_{\mathring{K}} \bar{\gamma}_{ij}  -  \tfrac{1}{3}\beta_{\mathring{K}} \gamma \mathring{h}_{ij}=C_{ij}- \tfrac{2}{3}\mathring{K} \sigma_{ij}+2 \mathring{L}(\xi)_{ij}. \label{LinearisedEqMetricSimplified3} 
%\end{equation}  
Taking the trace of \eqref{LinearisedEqMetricSimplified3} one obtains 
\[
-(\mathring{\Delta}+\beta)\gamma=C_k{}^k,
\]
which admits a unique solution, since
$\beta\notin\text{Spec}(-\mathring{\Delta})$ implies
that $-(\mathring{\Delta}+\beta)$ is
invertible. Substituting
into \eqref{LinearisedEqExtrinsicCurv} yields
\[
\mathring{\mathcal{D}}^*\circ\mathring{\mathcal{D}}(\bmvarsigma)_{ij}=A_{ij}+\mathring{\mathcal{D}}^*(\mathring{\star} \mathring{L}(\bar{\bmxi}))_{ij}+\tfrac{1}{9}\mathring{\mathcal{D}}^*\circ\mathring{\mathcal{D}}(\gamma \mathring{h}_{ij})
\]
where $\gamma$ is as determined in the previous step, for which there
exists a unique solution $\varsigma_{ij}$, since
$\mathring{\mathcal{D}}^*\circ\mathring{\mathcal{D}}:\mathscr{S}^2_0(\mathcal{S};\mathring{\bmh})\rightarrow
\mathscr{S}^2_0(\mathcal{S};\mathring{\bmh})$ is an
isomorphism. Finally, substituting the $\gamma$ and $\varsigma_{ij}$
so obtained into \eqref{LinearisedEqMetricSimplified3}, one obtains
\[
\mathring{\Delta}_L\bar{\gamma}_{ij} +4\bar{\gamma}_{ij}=  C_{ij}+2
\mathring{L}(\xi)_{ij}+\tfrac{1}{3}\beta \gamma
\mathring{h}_{ij}- \tfrac{2}{3}\mathring{K} \varsigma_{ij},
\]
which admits a unique solution since $(\mathring{\Delta}_L+4)$ is an isomorphism.
% Finally, substituting into \eqref{LinearisedEqMetricSimplified3}, 
%\begin{equation}
%\mathring{\Delta}_L\bar{\gamma}_{ij} - \alpha_K \bar{\gamma}_{ij} =C_{ij}-\tfrac{1}{3}C_k{}^k \mathring{h}_{ij}- \tfrac{2}{3}\mathring{K} \sigma_{ij}+2 \mathring{L}(\xi)_{ij},
%\end{equation}
%for which there exists a unique solution $\bar{\gamma}_{ij}$, since
%by assumption
%$\alpha_{\mathring{K}}\notin\text{spec}(\mathring{\Delta}_L)$. 

\medskip
The previous two steps conclude the proof that $D\Psi$ is an
isomorphism, and so by the Implicit Function Theorem there exists a
map $\nu$ from the freely-prescribed data to the space of solutions of
the auxiliary system $\Psi=0$. It only remains to establish the
injectivity of the map $\nu$.

\medskip
\noindent
\textbf{Injectivity of} $\mathbf{\nu}$. To establish the injectivity
of $\nu$, we need to consider the linearisation of $\Psi$ in the
direction of the free data ---namely
\[
\frac{\mbox{d}}{\mbox{d}\tau}\Psi(\bmchi,\bar{\bmX},\bmX,\bmh;~\mathring{K}+\tau\phi,\tau\bar{\bmT},\tau
\bmT) \bigg\vert_{\tau=0}=0.
\]
Since the background initial data, being hyperbolic, has umbilical extrinsic curvature, the expression
\eqref{LinearisationInDirectionOfFreeData} simplifies to
\begin{subequations}
\begin{eqnarray}
&& \mathring{L}(d\phi)_{jk}+3 \mathring{\epsilon}_{kil} \mathring{D}^{l}\bar{T}_{j}{}^{i} +3\mathring{\epsilon}_{jil} \mathring{D}^{l}\bar{T}_{k}{}^{i}=0,\label{LinearisationFreeData1}\\
&& \mathring{D}^{i}\bar{T}_{il} =0,\label{LinearisationFreeData2}\\
&& \mathring{D}^{j}T_{ij}=0,\label{LinearisationFreeData3}\\
&& T_{ij} - \tfrac{4}{9}\mathring{K}\phi \mathring{h}_{ij}  =0.\label{LinearisationFreeData4}
\end{eqnarray}
\end{subequations}
First consider the case $\mathring{K}\neq 0$: taking the trace of the
algebraic equation
\eqref{LinearisationFreeData4} one finds that $\phi=0$, and so
$T_{ij}=0$. Combining
\eqref{LinearisationFreeData1}--\eqref{LinearisationFreeData2} ---see Remark \ref{Remark:DecompositionOfCodazziOperator}--- and
using $\phi=0$, one obtains
\[
(\mathring{\mathcal{D}}\bar{T})_{ijk}\equiv\mathring{D}_i\bar{T}_{jk}-\mathring{D}_j\bar{T}_{ik}=0. 
\]
Now, we have assumed the non-existence of
non-trivial tracefree Codazzi tensors, so $\bar{T}_{ij}=0$. Hence, in
the non--time symmetric case $\mathring{K}\neq 0$, the map $\nu$ is
injective.

\smallskip
Consider on the other hand the time-symmetric case
$\mathring{K}=0$. Clearly, the kernel of the system contains triples of
the form
\begin{equation}
\label{KernelOfnuInTSCase}
(T_{ij},\bar{T}_{ij},\phi)=(\bm0,~\bm0,~\text{const.}). 
\end{equation}
We show that these are indeed the only solutions. First, note that
condition \eqref{LinearisationFreeData4} (setting $\mathring{K}=0$)
again implies $T_{ij}=0$. Now, taking the divergence of
\eqref{LinearisationFreeData1}, one has that  
\begin{align*}
0&=\mathring{\delta} \mathring{L}(d\phi)_k+3\mathring{\epsilon}_{kil} \mathring{D}^{j}\mathring{D}^{l}\bar{T}_{j}{}^{i} + 3 \mathring{\epsilon}_{jil} \mathring{D}^{j}\mathring{D}^{l}\bar{T}_{k}{}^{i}\\
&=\mathring{\delta}\mathring{L}(d\phi)_k+\tfrac{3}{2} \mathring{\epsilon}^{jlm} \bar{T}_{k}{}^{i} \mathring{r}_{ijlm} -  \tfrac{3}{2} \mathring{\epsilon}_{i}{}^{lm} \bar{T}^{ij} \mathring{r}_{kjlm} + 3 \mathring{\epsilon}_{kjl} \mathring{D}_{i}\mathring{D}^{l}\bar{T}^{ij}\\
&=\mathring{\delta} \mathring{L}(d\phi)_k+6 \mathring{\epsilon}_{kjl} \bar{T}^{ij} \mathring{r}_{i}{}^{l} + 3 \mathring{\epsilon}_{kjl} \mathring{D}^{l}\mathring{D}_{i}\bar{T}^{ij}\\
&=\mathring{\delta} \mathring{L}(d\phi)_k,
\end{align*}
after commuting covariant derivatives and 
where in the last step we are using the fact that the background
metric is Einstein, along with the fact that $\bar{T}_{ij}$ is
divergence-free. Integrating by parts, one then finds that
$\mathring{L}(d\phi)=0$ ---that is to say, $d\phi$ is a conformal
Killing vector. Since $\mathring{\bmh}$ admits no non-trivial
conformal Killing vectors, $d\phi=0$ and so
$\phi$ is constant. Proceeding as in the $\mathring{K}\neq 0$ case, we
again see that $\bar{T}_{ij}=0$, as a consequence of there being no
non-trivial tracefree Codazzi tensors. By restricting the choice of
$\phi$ to the sub-Banach space of functions integrating to zero,
we clearly exclude from the kernel triples of the form \eqref{KernelOfnuInTSCase},
ensuring that $\nu$ is injective. 

\smallskip
In order to show that $u\mapsto w(u)$ is injective, all that
remains to be shown is that the map $u\equiv
(\phi,\bm{T},\bar{\bmT})\mapsto \bmS(\bmX(u),\bmT)$ is injective (and
likewise for $\bar{\bmX}$).
%\[X_i\rightarrow \mathring{L}( X(T))_{ij}+T_{ij}\xrightarrow{
%\text{proj}_h} S(\bmX(\bmT),\bmT)_{ij}\] 
The injectivity of the map
$u\mapsto \mathring{L}(\bmX(u))+\bmT$ follows from injectivity of
$\nu$ and uniqueness of the York split ---using, once again, the
non-existence of conformal Killing vectors for $\mathring{\bmh}$, see
\cite{Can81}. Finally, we need to show that $\Pi_{\bmh}$ is
injective (for $\bmh$ sufficiently close to $\mathring{\bmh}$ in
$\mathcal{B}_\bmh$). 
% Intuitively, this follows from the observations
% that $\text{proj}_{\mathring{h}}$ is just the identity on
% $\mathscr{S}^2_0(\mathcal{S};\mathring{h})$, and that $\text{proj}_h$
% is continuously-dependent on the metric $\bmh$, so that for a
% sufficiently small metric perturbation $\text{proj}_h$ inherits
% injectivity from the identity map.
To see this, note that if
$T_{ij}\in\text{ker}(~\Pi_\bmh)\cap\mathscr{S}^2_0(\mathcal{S};\mathring{\bmh})$,
then
\[
T_{ij}=\tfrac{1}{3}Th_{ij} 
\] 
with $T=\text{tr}_\bmh(\bmT)$, and
\[
0=T\cdot\text{tr}_{\mathring{\bmh}}\bmh=T\cdot
(3+\text{tr}_{\mathring{\bmh}}(\bmh-\mathring{\bmh})). 
\]
Now, by Sobolev Embedding (see \cite{LawMic89}) the $C^0-$norm of $(\bmh-\mathring{\bmh})$ is bounded above by the $H^2-$norm and hence, for $\bmh$ sufficiently close to $\mathring{\bmh}$ in $\mathcal{B}_\bmh$, it follows that $T=0$ and hence $T_{ij}=0$ ---that is to say,
$\Pi_\bmh$ is injective for such an $\bmh$.

%Finally, since the background initial data is smooth, the map $D\Psi$ is a smooth map of Banach spaces, and hence by the Implicit Function Theorem, the map $\nu$ is also smooth. Therefore, the map $w$ (see the statement of Theorem \ref{MainTheorem}) maps smooth free data $(\phi,\bmT,\bar{\bmT})$ to smooth solutions of $\Psi=0$. 
%\mnotex{Have commented out smoothness stuff}
%\begin{align*}
%\Vert \text{proj}_h\Vert_{\text{op}}^2 &\geq \tfrac{1}{2}\Vert \text{proj}_{\mathring{h}}\Vert_{\text{op}}^2-\Vert \text{proj}_h-\text{proj}_{\mathring{h}}\Vert_{\text{op}}^2 \geq\tfrac{1}{2}-C\Vert \bmh-\mathring{\bmh}\Vert_{L^2}^2
%\end{align*} 
%for some constant $C$, since $\text{proj}_h$ is continuously-dependent on $\bmh$.
% since $\text{proj}_h$ is continuously-dependent on $\bmh$, and using the fact that $\Vert \text{proj}_{\mathring{h}}\Vert_{\text{op}}=1$, since 
%\[\text{proj}_{\mathring{h}}:\mathscr{S}^2_0(\mathcal{S};\mathring{h})\rightarrow \mathscr{S}^2_0(\mathcal{S};\mathring{h})\]
%is just the identity map.
%Hence, for sufficiently small $\Vert \bmh-\mathring{\bmh}\Vert_{L^2}$, it follows that $\Vert \text{proj}_h\Vert_{\text{op}} >0$ and hence that $\text{proj}_h$ is injective.

\end{proof}

% \begin{remark}
% {\em Here, $\Vert \text{proj}_h\Vert_{\text{op}}$ denotes the operator
% norm of $\text{proj}_h$ as a map from
% $(\mathscr{S}^2_0(\mathcal{S};\mathring{h}),\Vert\cdot\Vert_{L^2(\mathring{h})})$
% to itself ---i.e.
% \[
% \text{sup}_{\bmT\in\mathscr{S}^2_0(\mathcal{S};\mathring{h})\setminus\lbrace
% 0\rbrace} ~\frac{\Vert \text{proj}_h(T)\Vert_{L^2}}{\Vert
% T\Vert_{L^2}},
% \] with the $L^2-$norms defined again with respect to the background
% metric, $\mathring{\bmh}$.}
% \end{remark}

%\begin{remark}{\em Note that the proof of the injectivity and surjectivity of $D\Psi$ is significantly simpler in the case of $\mathring{K}=0$ (i.e. time-symmetry of the background solution), since in this case equations ()--() and ()--() reduce to
%\[ ,\]
%which are semi-decoupled ---solve () first to get $\sigma_{ij}=0$...
%}
%\end{remark}

\begin{remark}
{\em Recall the notion of \textit{total mean extrinsic curvature}
\[
\int_{\mathcal{S}}\text{tr}_{\mathring{\bmh}}(\bmK)~d\mathring{\mu},
\]
given here with respect to the background metric 
$\mathring{\bmh}$. The additional requirement that $\phi$ integrates to zero 
in the time-symmetric case $\mathring{K}=0$
therefore ensures that the corresponding solutions furnished by
Theorem \ref{MainTheorem} have zero total mean extrinsic curvature
with respect to $\mathring{\bmh}$. While the proof guarantees a solution
for any choice of (smooth, sufficiently small) $\phi$, the injectivity
of the map $\nu$ is only guaranteed if we further restrict to those $\phi$ 
that integrate to zero.}
\end{remark}
% \mnotex{Juan was asking about the $\lambda=0$ case (corresponding to
% $\beta_{\mathring{K}}=4$). To see whether the theorem allows for
% $\lambda=0$ we need determine whether
% $4\in\text{spec}(-\mathring{\Delta}$ ---I can't find any references to
% eigenvalues of the scalar Laplacian in spaces of negative Ricci
% curvature.}

% \mnotex{Confusion: is $\phi$ the trace or the difference of the trace
% from the $\mathring{K}$?! I have been inconsistent. I think I've
% corrected this --now consistent with the rough version of the theorem
% at the beginning, I think.}

\begin{remark}
{\em In the proof of Proposition
\ref{PropExistenceOfSolutionsToSecondarySystem}, we could have instead
used the vanishing of the index to establish surjectivity. Recall that
the Atiyah--Singer index theorem (see \cite{Nak03}, for example)
relates the analytical and topological index of an elliptic operator
over a compact manifold. For an odd-dimensional base manifold
$\mathcal{S}$ the topological index vanishes ---see the discussion in
\cite{Nak03}--- and so the index theorem guarantees that an injective
elliptic operator defined over an odd-dimensional manifold must in
fact be an isomorphism of the appropriate Banach spaces.}
\end{remark}

\subsection{Sufficiency of the auxiliary system}
\label{Section:Sufficiency}
In this section we establish \emph{sufficiency} of
auxiliary constraint system ---that is, we show that the solutions of
the auxiliary system established in the previous section are indeed solutions
of the extended constraint equations.

\subsubsection{Injectivity of $\mathcal{K}_{\bmh}$}
\label{Subsec:PreliminaryRemarksSufficiency}

Recall the operator $\mathcal{K}_{\bmh}$ (see Section \ref{Subsubsection:EllipticEquationsForQandJ}) given by 
\[\mathcal{K}_{\bmh}(\bmJ)=\left(\begin{array}{c}
\mathring{\mathcal{D}}^*(\bmJ)_{ij}\\
\epsilon^{ijk}D_iJ_{jkl}
\end{array}\right). \]
As described in Section \ref{Subsection:SufficiencyArgument}, the
sufficiency argument will involve establishing injectivity of the
operator $\mathcal{K}_{\bmh}$. We first consider the operator
evaluated at the background metric, $\mathring{\bmh}$:

\begin{proposition}
\label{Prop:BackgroundSufficiencyJ}
Let $(\mathcal{S},\mathring{\bmh})$ be a smooth conformally rigid
hyperbolic manifold, then the operator $\mathring{\mathcal{K}}\equiv
\mathcal{K}_{\mathring{\bmh}}$ is injective ---i.e. the system of
equations $\mathring{\mathcal{K}}(\bmJ)=0$ admits only the trivial
solution $J_{ijk}=0$.
\end{proposition}

\begin{proof}
%
%\begin{subequations}
%\begin{eqnarray}
%&& \mathring{\mathcal{D}}^*(J)_{ij}=0,\label{Eq:AuxiliaryEqForExtrinsicCurv}\\
%&& \mathring{\epsilon}^{ijk}\mathring{D}_iJ_{jkl}=0, \label{Eq:BackgroundIntegrabilityCondition}
%\end{eqnarray}
%\end{subequations}
Suppose $J_{ijk}=0$ is a Jacobi tensor satisfying $\mathring{\mathcal{K}}(\bmJ)=0$. Performing the Jacobi decomposition of $J_{ijk}$ with respect to $\mathring{\bmh}$ we obtain 
\begin{subequations}
\begin{eqnarray}
&& 2\mathring{\text{rot}}_2(\bmF)_{ij}+\mathring{L}(\bmA)_{ij}=0,\label{Eq:AuxEqDecomposed}\\
&& \mathring{\delta}(\bmF)_i+\mathring{\text{curl}}(\bmA)_i=0,\label{Eq:BackgroundIntegrDecomposed}
\end{eqnarray}
\end{subequations}
with $\mathring{\text{curl}}(\bmA)_i\equiv \mathring{\epsilon}_{ijk}\mathring{D}^jA^k$, to be read as equations for $F_{ij}\in\mathscr{S}^2_0(\mathcal{S};\mathring{\bmh})$ and $A_i\in\Lambda^1(\mathcal{S})$. It then follows that
\begin{align*}
0&=\mathring{\delta}(\mathring{L}(\bmA)+2\mathring{\text{rot}}_2(\bmF))_i\\
&=\mathring{\delta}\circ \mathring{L}(\bmA)_i+2\mathring{\delta}\circ\mathring{\text{rot}}_2(\bmF)_i\\
&=\mathring{\delta}\circ \mathring{L}(\bmA)_i+\mathring{\text{curl}}\circ\mathring{\delta}(\bmF)_i-2\mathring{\epsilon}_{iml}\mathring{r}_j{}^lF^{jm}\\
&=\mathring{\delta}\circ \mathring{L}(\bmA)_i-\mathring{\text{curl}}^2(\bmA)_i-2\mathring{\epsilon}_{iml}\mathring{r}_j{}^lF^{jm},
\end{align*}
where the first line follows from \eqref{Eq:AuxEqDecomposed}, the third uses the identity 
\[\mathring{\delta}\circ\mathring{\text{rot}}_2(\bmF)_i=\tfrac{1}{2}\mathring{\text{curl}}\circ\mathring{\delta}(\bmF)_i-\mathring{\epsilon}_{iml}\mathring{r}_j{}^lF^{jm}, \]
and the fourth follows from substitution using \eqref{Eq:BackgroundIntegrDecomposed}. Since $\mathring{\bmh}$ is Einstein, we find
\[\mathring{\delta}\circ \mathring{L}(\bmA)_i-\mathring{\text{curl}}^2(\bmA)_i=0.\]
Contracting with $A^i$ and integrating by parts:
\begin{equation}
\label{Eq:IdentityForCurlyPOperatorInEinsteinCase}
0=\int_{\mathcal{S}}\left(\tfrac{1}{2}\Vert \mathring{L}(\bmA)\Vert^2+\Vert\mathring{\text{curl}}(\bmA)\Vert^2\right)~d\mu_{\mathring{\bmh}}, 
\end{equation}
where we are using the fact that $\mathring{\delta}^*=-\tfrac{1}{2}\mathring{L}$ and $\mathring{\text{curl}}^*=\mathring{\text{curl}}$. Hence, we find that $A_i=0$, since $\mathring{\bmh}$ admits no conformal Killing vector fields. Substituting into \eqref{Eq:AuxEqDecomposed}--\eqref{Eq:BackgroundIntegrDecomposed}, we see that $\mathring{\text{rot}}_2(\bmF)_{ij}=\mathring{\delta}(\bmF)_i=0$ and hence $F_{ij}=0$ since $\mathring{\bmh}$ admits no tracefree Codazzi tensors. It follows then that $J_{ijk}=0$.
\end{proof}
In order to show that $\mathcal{K}_{\bmh}$ is injective for $\bmh$
sufficiently close to $\mathring{\bmh}$, we will first show that the
operator $\mathcal{K}_{\bmh}$ is elliptic and then appeal to a
particular stability property of the kernel of elliptic
operators. Let us first establish ellipticity:

\begin{lemma}
\label{Lemma:EllipticityOfCurlyKOperator}
The operator $\mathcal{K}_{\bmh}$ is first-order elliptic for any Riemannian metric $\bmh$. 
\end{lemma}
\begin{proof}
Recall from Remark \ref{Remark:JacobiDecomp} that
$\mathcal{J}(\mathcal{S})$ and
$\mathscr{S}^2_0(\mathcal{S};\mathring{\bmh})\oplus\Lambda^1(\mathcal{S})$ are
isomorphic as vector spaces. Therefore, in order to establish
ellipticity it suffices to show that $\mathcal{K}_{\bmh}$ is
overdetermined elliptic. Note that the second component of $\mathcal{K}_{\bmh}=0$ is equivalent to 
\[D_{[i}J_{jk]l}=0. \]
Note also that a change of connection $D_i\rightarrow\mathring{D}_i$ only introduces lower-order (i.e. algebraic) terms involving $J_{ijk}$, so in order to show ellipticity it suffices to consider the operator $\mathring{\mathcal{K}}$, or equivalently an operator with principal part
\[\left(\begin{array}{c}
\mathring{\mathcal{D}}^*(\bmJ)_{ij}\\
\mathring{D}_{[i}J_{jk]l}.
\end{array}\right).\]
Accordingly, suppose
$J_{ijk}\in\mathcal{J}(\mathcal{S})$ is in the kernel of the symbol
map, $\sigma_\xi [\mathring{\mathcal{K}}]$, for a given fixed
$\xi_i$, so that
\begin{subequations}
\begin{eqnarray}
&& \xi^k J_{ikj}+\xi^k J_{jki}-\tfrac{2}{3}\xi^k J_{lk}{}^l\mathring{h}_{ij}=0,\label{SymbolOfQVanishes1}\\
&& \xi_i J_{jkl}+\xi_j J_{kil}+\xi_k J_{ijl}=0.\label{SymbolOfQVanishes2}
\end{eqnarray}
\end{subequations}

Note that the latter is indeed equivalent to $\epsilon^{ijk}\xi_i
J_{jkl}=0$, taking into account the fact that
$J_{ijk}=-J_{jik}$. Contracting indices $i,l$ in equation
\eqref{SymbolOfQVanishes2}, we obtain
\begin{equation}
\xi^l J_{jkl}=-\xi_j J_{kl}{}^l+\xi_kJ_{jl}{}^l.\label{SymbolOfQVanishes2Contracted}
\end{equation}
On the other hand, contracting \eqref{SymbolOfQVanishes1} with $\xi^j$, we obtain
\begin{align}
0&=\xi^k\xi^j J_{ikj}+\xi^k\xi^j J_{jki}-\tfrac{2}{3}\xi^k\xi_i J_{lk}{}^l\nonumber\\
&=\xi^k\xi^j J_{ikj}-\tfrac{2}{3}\xi^k\xi_iJ_{lk}{}^l\nonumber\\
&=\tfrac{1}{3}\xi_i\xi^kJ_{kl}{}^l+\vert{\bm \xi}\vert^2J_{il}{}^l\label{SigmaTimesSymbolOfQVanishes1}
\end{align}
where the second line follows from the fact that $J_{ijk}=-J_{jik}$
and the third line follows from substituting
\eqref{SymbolOfQVanishes2Contracted}. Contracting
\eqref{SigmaTimesSymbolOfQVanishes1} with $\xi^i$, we find that
$\xi^iJ_{il}{}^l=0$, which when substituted back into
\eqref{SigmaTimesSymbolOfQVanishes1} yields
$J_{il}{}^l=0$. Substituting the latter into
\eqref{SymbolOfQVanishes1} and \eqref{SymbolOfQVanishes2Contracted} we
see that
\begin{equation}
\xi^k J_{ikj}+\xi^k J_{jki}=0\label{SymmetrisedSigmaJContractionEqualsZero}
\end{equation} 
in addition to $\xi^k J_{ijk}=0$. If we instead contract indices $k,l$ in \eqref{SymbolOfQVanishes2}, we obtain 
\begin{align*}
0&=\xi^kJ_{ijk}+\xi^kJ_{jki}+\xi^k J_{kij}\\
&=\xi^k J_{jki}-\xi^k J_{ikl}
\end{align*}
where the second line follows from $\xi^k J_{ijk}=0$ and the fact that
$J_{ijk}=-J_{jik}$. Combining with
\eqref{SymmetrisedSigmaJContractionEqualsZero}, we find that $\xi^k
J_{ikj}=\xi^k J_{kij}=0$. Finally, contracting
\eqref{SymbolOfQVanishes2} with $\xi^i$, we obtain
\[
0=\vert\bm\xi\vert^2J_{jkl}+\xi_j\xi^iJ_{kil}+\xi_k\xi^iJ_{ijl}=\vert\bm\xi\vert^2J_{jkl}
\]
where the second equality follows from $\xi^k J_{ikj}=\xi^k
J_{kij}=0$. Hence, for $\xi_i\neq 0$, we see that the symbol map is
injective ---that is to say, $\mathcal{K}_{\bmh}$ is overdetermined
elliptic and hence determined elliptic, since its domain and codomain
 are of equal dimension as vector spaces.
\end{proof}

In order to establish injectivity of $\mathcal{K}_{\bmh}$ we will make
use of an elliptic estimate. Rather than working directly with the
first-order operator $\mathcal{K}_{\bmh}$ we choose instead to work
with the elliptic operator
$\mathcal{K}_{\bmh}^*\circ\mathcal{K}_{\bmh}$ to which the more
standard results of second-order elliptic operators may be applied
---note that the kernel of the latter operator agrees with the kernel
of $\mathcal{K}_{\bmh}$, so it suffices to show injectivity of the
second-order operator. Our starting point is the following elliptic
estimate for $\mathring{\mathcal{K}}^*\circ\mathring{\mathcal{K}}$:
there exists $C>0$ such that, for all
$\bm\eta\in H^2(\mathcal{J}(\mathcal{S}))$
\begin{equation}
\label{EllipticEstimate}
\Vert \bm\eta\Vert_{H^2}\leq C\left(\Vert \mathring{\mathcal{K}}^*\circ\mathring{\mathcal{K}}(\bm\eta)\Vert_{L^2}+\Vert \bm\eta\Vert_{H^1}\right)
\end{equation}
---see Appendix II of \cite{Cho08}, for instance. In fact, we will
require a \emph{uniform} version of the above elliptic estimate which
allows for small perturbations of the metric:

\begin{lemma}
\label{Lemma:UniformEstimate}
There exists $\varepsilon>0$ such that, for all $\bmh$ satisfying $\Vert
\bmh-\mathring{\bmh}\Vert_{H^s}<\varepsilon $, $s\geq 4$, we have the 
estimate
\begin{equation}
\label{UniformEllipticEstimate}
 \Vert\bm\eta\Vert_{H^2} \leq 2C\left(\Vert \mathcal{K}_{\bmh}^*\circ\mathcal{K}_{\bmh}(\bm\eta)\Vert_{L^2}+\Vert \bm\eta\Vert_{H^1}\right) 
\end{equation}
for all $\bm\eta\in H^2(\mathcal{J}(\mathcal{S}))$, with $C$ as in
\eqref{EllipticEstimate}, depending only on $\mathring{\bmh}$.
\end{lemma}

\begin{proof}
We first note that there exists some constant $\tilde{C}$ such that
for any given $\bm\eta\in\mathcal{J}(\mathcal{S})$, we have
\begin{equation}
\label{LipschitzProperty}
\Vert (\mathcal{K}^*_{\bmh}\circ\mathcal{K}_{\bmh}-\mathring{\mathcal{K}}^*\circ\mathring{\mathcal{K}})\bm\eta\Vert_{L^2}\leq \tilde{C}\Vert \bmh-\mathring{\bmh}\Vert_{H^2}\Vert \bm\eta\Vert_{H^2} 
\end{equation}
 ---this follows from the fact that, schematically,  
\[
(\mathcal{K}^*_{\bmh}\circ\mathcal{K}_{\bmh}-\mathring{\mathcal{K}}^*\circ\mathring{\mathcal{K}})\bm\eta\sim
(\bmh-\mathring{\bmh})\partial\partial\bm\eta+\bmS\cdot \partial\bm\eta+(\partial\bmS+\bmS\cdot\bmS)\bm\eta 
\] 
with $\bmS$ the transition tensor covariant derivatives associated to
the metrics $\mathring{\bmh}$ and $\bmh$, from which it is clear then that
$(\mathcal{K}^*_{\bmh}\circ\mathcal{K}_{\bmh}-\mathring{\mathcal{K}}^*\circ\mathring{\mathcal{K}})\bm\eta$
may be bounded above by $\Vert \bmh-\mathring{\bmh}\Vert_{H^2}\Vert
\bm\eta\Vert_{H^2}$.
\\

Now, using inequality 
\eqref{LipschitzProperty} we find that for all $\bmh$ satisfying $\Vert
\bmh-\mathring{\bmh}\Vert_{H^2}<\varepsilon$, and for all $\bm\eta\in \mathcal{J}(\mathcal{S})$,  
\begin{align*}
\Vert \bm\eta\Vert_{H^2}&\leq C\left(\Vert \mathring{\mathcal{K}}^*\circ\mathring{\mathcal{K}}(\bm\eta)\Vert_{L^2}+\Vert \bm\eta\Vert_{H^1}\right)\\
 &\leq C\left(\Vert \mathcal{K}_{\bmh}^*\circ\mathcal{K}_{\bmh}(\bm\eta)\Vert_{L^2}+\Vert (\mathring{\mathcal{K}}^*\circ\mathring{\mathcal{K}}-\mathcal{K}_{\bmh}^*\circ\mathcal{K}_{\bmh})\bm\eta\Vert_{L^2}+\Vert \bm\eta\Vert_{H^1}\right)\\
 &\leq C\left(\Vert \mathcal{K}_{\bmh}^*\circ\mathcal{K}_{\bmh}(\bm\eta)\Vert_{L^2}+\varepsilon \tilde{C}\Vert \bm\eta\Vert_{H^2} +\Vert \bm\eta\Vert_{H^1}\right),
\end{align*}
with $C$ depending only on $\mathring{\bmh}$. Thus, taking
$\varepsilon=1/(2C\tilde{C})$ and rearranging we have that 
\begin{equation}
\label{UniformEllipticEstimate}
 \Vert\bm\eta\Vert_{H^2} \leq 2C\left(\Vert \mathcal{K}_{\bmh}^*\circ\mathcal{K}_{\bmh}(\bm\eta)\Vert_{L^2}+\Vert \bm\eta\Vert_{H^1}\right) 
\end{equation}
for all $\bm\eta\in H^2(\mathcal{J}(\mathcal{S}))$ and for all $\Vert \bmh-\mathring{\bmh}\Vert_{H^2}<\varepsilon $ as required.
\end{proof}

\begin{remark}
{\em The content of inequality \eqref{LipschitzProperty} may be summarised by the statement that the map
\[ 
\begin{array}{cccc}
M: & H^2(\mathscr{S}^2(\mathcal{S})) & \longrightarrow & B(H^2(\mathcal{J}(\mathcal{S})),L^2(\mathcal{J}(\mathcal{S}))\\
& \bmh & \longmapsto & \mathcal{K}_{\bmh}^*\circ\mathcal{K}_{\bmh}
\end{array}
\]
is Lipschitz continuous at $\bmh=\mathring{\bmh}$ ---here,
$B(\cdot,\cdot)$ denotes the Banach space of bounded linear maps
between the indicated Banach spaces, endowed with the operator norm---
with $\tilde{C}$ the Lipschitz constant, which depends on the precise
structure of $\mathcal{K}^*\circ\mathcal{K}$ and may be computed
explicitly.}
\end{remark}

%\begin{remark}{\em Of course, taking a smaller $\epsilon$ in the above proposition leads to a stronger elliptic estimate, but this will not be necessary for our purposes.}
%\end{remark}

\subsubsection{The main argument}
\label{Subsubsection:SufficiencyMainArgument}

Assume now that the procedure described in Section \ref{Section:ExistenceOfSolutions} has
been carried out ---that is to say, we have established the existence
of a neighbourhood of solutions to the auxiliary system. For each such
solution, the corresponding zero quantities $Q_i,~J_{ijk}$ necessarily
satisfy
\begin{subequations}
\begin{eqnarray}
&& \mathcal{K}_\bmh(\bmJ)=0, \label{AuxiliaryEqFrExtrCurvInZeroQuantities}\\
%&& \epsilon^{ijk}D_i J_{jkl}=0, \label{IntegrabilityConditionExtrCurvALT}\\
&& D^{i}(\mathcal{L}_Q \bmh)_{ij}-\tfrac{1}{2}D_j(\mathcal{L}_Q \bmh)_i{}^i =  K_{ik} J_{j}{}^{ik} -
K_{jk}J^{ik}{}_{i}  - K J_{j}{}^{i}{}_{i}. \label{IntegrabilityConditionForQ}
\end{eqnarray}
\end{subequations}
The first equation collects together \eqref{HavingSolvedExtrinsicCurvAuxiliaryEq} and \eqref{IntegrabilityConditionExtrCurv}, while the latter is the remaining integrability condition
 --- see Section \ref{Subsection:SufficiencyArgument}. We regard the
above as equations for a pair of tensor fields $\bmQ\in\Lambda^1(\mathcal{S}),~\bmJ\in\mathcal{J}(\mathcal{S})$, which we
aim to prove are necessarily vanishing ---at this point we forget about the 
definitions of the zero quantities $Q_i,~J_{ijk}$ in terms of the unknown tensor fields. 
\\

We first use the results of the previous section to show that
injectivity of the operator $\mathcal{K}_{\bmh}$ is stable under
$H^s$-perturbations, $s\geq 4$, of the metric. Note that, in the following, 
all Sobolev norms are taken with respect to the background metric, $\mathring{\bmh}$.
% \mnotex{General comment about norms added. I think it's true that
% we're using the background norms, right?}

\begin{proposition}
\label{Prop:StabilityCodazziOperator}
There exists $\varepsilon>0$ such that for any metric $\bmh$ satisfying
$\Vert \bmh-\mathring{\bmh}\Vert_{H^s}<\varepsilon $, the corresponding
operator
$\mathcal{K}_{\bmh}$ is injective in $H^2$.
\end{proposition}

%\mnotex{Do we need the metrics to be smooth? Fix the $s$.}

\begin{proof}
Suppose not. Then there exists a \emph{failure sequence}
$\{(\bmh^{(n)},~\bm\eta^{(n)})\}$, $n\in\mathbb{N}$ ---i.e. a sequence of
Riemannian metrics $\bmh^{(n)}$ converging to $\mathring{\bmh}$ in
$H^2$ and corresponding non-zero Jacobi tensors
$\bm\eta^{(n)}\in\mathcal{J}(\mathcal{S})$ for which
\[
\mathcal{K}_{(n)}(\bm\eta^{(n)})=0 
\]
for each $n\in\mathbb{N}$ ---here, $\mathcal{K}_{(n)}\equiv
\mathcal{K}_{\bmh^{(n)}}$. Since $\mathcal{K}_{(n)}$ is linear, we may
take each $\bm\eta^{(n)}$ to be of unit $H^2$-norm.
Hence, by the Rellich-Kondrakov Theorem, since the sequence
$\{\bm\eta^{(n)}\}$ is bounded in $H^2$, there is a subsequence
that is Cauchy in $H^1$ ---let us assume without loss of generality
that $\{\bm\eta^{(n)}\}$ is Cauchy--- converging to some limit
$\bm\eta^\bullet\in\mathcal{J}(\mathcal{S})$. We
now aim to show using the inequality \eqref{UniformEllipticEstimate} that the
sequence is in fact Cauchy in $H^2$. Let us restrict
to a the tail of the subsequence
(relabelling, if necessary) for which $\Vert
\bmh^{(n)}-\mathring{\bmh}\Vert<\varepsilon $ with $\varepsilon$ as given in
Proposition \ref{Lemma:UniformEstimate}. Applying the inequality 
\eqref{UniformEllipticEstimate} to
$\bm\eta^{(m,n)}\equiv\bar{\bm\eta}^{(n)}-\bar{\bm\eta}^{(m)}$, with
$\bmh=\bmh^{(n)}$, we have
\begin{align}
&\hspace{-3mm}\Vert \bm\eta^{(m,n)}\Vert_{H^2}\nonumber\\
&\leq 2C \left(\Vert \mathcal{K}^*_{(n)}\circ\mathcal{K}_{(n)}(\bm\eta^{(m,n)})\Vert_{L^2}+\Vert\bm\eta^{(m,n)}\Vert_{H^1}\right)\nonumber\\
&=2C \left(\Vert \mathcal{K}^*_{(n)}\circ\mathcal{K}_{(n)}(\bm\eta^{(m)}) \Vert_{L^2}+\Vert \bm\eta^{(m,n)}\Vert_{H^1}\right)\nonumber\\
&= 2C \left(\Vert (\mathcal{K}^*_{(n)}\circ\mathcal{K}_{(n)}-\mathcal{K}^*_{(m)}\circ\mathcal{K}_{(m)})\bm\eta^{(m)} \Vert_{L^2}+\Vert \bm\eta^{(m,n)}\Vert_{H^1}\right).\label{ShowingSequenceIsCauchy}
\end{align}    
The second line follows from by substituting for $\bm\eta^{(m,n)}$ in
the first term and using the fact that, by assumption,
$\mathcal{K}_{(n)}(\bar{\bmeta}^{(n)})=0$; the third line follows
similarly. Now,
\begin{multline*}
\Vert(\mathcal{K}^*_{(n)}\circ\mathcal{K}_{(n)}-\mathcal{K}^*_{(m)}\circ\mathcal{K}_{(m)})\bm\eta^{(m)}\Vert_{L^2}\leq \Vert(\mathcal{K}^*_{(n)}\circ\mathcal{K}_{(n)}-\mathring{\mathcal{K}}^*\circ\mathring{\mathcal{K}})\bm\eta^{(m)}\Vert_{L^2}\\
+\Vert(\mathcal{K}^*_{(m)}\circ\mathcal{K}_{(m)}-\mathring{\mathcal{K}}^*\circ\mathring{\mathcal{K}})\bm\eta^{(m)}\Vert_{L^2},
\end{multline*}
which goes to zero in the limit $m,n\longrightarrow\infty$, again
using the Lipschitz property of $M$ and the fact that
$\bm\eta^{(m)}$ is bounded in $H^2$. Collecting together the above
observations, we see from \eqref{ShowingSequenceIsCauchy} that as
$m,n\longrightarrow\infty$, $\bmeta^{(m,n)}\longrightarrow 0$ in $H^2$
---i.e. the sequence $\bar{\bm\eta}^{(n)}$ is Cauchy in $H^2$, and
therefore the limit
$\bmeta^\bullet\in\mathcal{J}(\mathcal{S})$ is in
$H^2$. Clearly $\bm\eta^\bullet$ is non-zero ---in fact, one has that $\Vert
\bm\eta^\bullet\Vert_{H^2}=1$.
\\

Using the Lipschitz property of $M$ once more, along with the fact
that $\bm\eta^{(n)}$ converges to $\bm\eta^\bullet$ in $H^2$, one
finds that 
\[
\Vert\mathring{\mathcal{K}}^*\circ\mathring{\mathcal{K}}(\bm\eta^\bullet)\Vert_{L^2}=\lim_{n\rightarrow\infty}\Vert\mathcal{K}^*_{(n)}\circ\mathcal{K}_{(n)}(\bm\eta^{(n)})\Vert_{L^2}=0. 
\]
Hence,
$\mathring{\mathcal{K}}^*\circ\mathring{\mathcal{K}}(\bm\eta^\bullet)=0$,
and it follows via integration by parts that
$\mathring{\mathcal{K}}(\bm\eta^\bullet)=0$. However,
$\bm\eta^\bullet\in\mathcal{J}(\mathcal{S})\setminus
\lbrace 0\rbrace$ and so we obtain a contradiction, since
$\mathring{\mathcal{K}}$ is injective, as shown in Proposition \ref{Prop:BackgroundSufficiencyJ}.
\end{proof} 

We are now in a position to prove the main result of this section:
 
\begin{proposition}[\textbf{\em Sufficiency}]
\label{PropositionSufficiency}
There exists an open neighbourhood $\mathcal{V}$ of
$\mathring{\bmh}\in\mathcal{B}_{\bmh}$, such that for each
$\bmh\in\mathcal{V}$, $(J_{ijk},Q_i)=(\bm0,\bm0)$ is the unique $H^2$
solution of
\eqref{AuxiliaryEqFrExtrCurvInZeroQuantities}--\eqref{IntegrabilityConditionForQ}.
\end{proposition} 

\begin{proof}
We begin by showing that $J_{ijk}=0$. This follows immediately from
the previous proposition provided we choose $\mathcal{V}$ to be a
suitably small neighbourhood. 
\\

Having established that $J_{ijk}=0$,
\eqref{IntegrabilityConditionForQ} implies that $Q_i$ satisfies the
integral identity \eqref{SufficiencyIntegralForQ}. Hence, it follows
that
\[
0=\int_{\mathcal{S}}\left(\Vert D\bmQ\Vert_\bmh^2-r_{ij}Q^iQ^j\right) ~d\mu_\bmh\geq \int_{\mathcal{S}}-r_{ij}Q^iQ^j ~d\mu_\bmh \longrightarrow \int_{\mathcal{S}} 2\Vert\bmQ\Vert^2_{\mathring{\bmh}} ~d\mu_{\mathring{\bmh}},
\]
where convergence follows from the fact that, since $\bmh\rightarrow \mathring{\bmh}$
in $H^4$, we have $r[\bmh]_{ij}\rightarrow \mathring{r}_{ij}=-2\mathring{h}_{ij}$ in $C^0$ ---convergence of the latter in $H^2$ is immediate, and an application of the Sobolev Embedding Theorem establishes convergence in $C^0$. Hence, provided we take
$\mathcal{V}$ to be a suitably-small neighbourhood, it follows that for any $\bmh\in\mathcal{V}$ we necessarily have
$\bmQ=0$.
\end{proof}

Hence, it follows that for solutions
$(K_{ij},{S}_{ij},\bar{S}_{ij},h_{ij})$ of the auxiliary system
sufficiently close to the background data, the corresponding zero quantities $Q_i,~J_{ijk}$ must
necessarily vanish, implying $(K_{ij},{S}_{ij},\bar{S}_{ij},h_{ij})$ 
indeed solves the extended constraint equations. This concludes the
proof of sufficiency. Collecting together Propositions
\ref{PropExistenceOfSolutionsToSecondarySystem} and
\ref{PropositionSufficiency}, one obtains Theorem \ref{MainTheorem}.

\begin{remark}
{\em Alternatively, we could also have shown $Q_i=0$ by using identity
\eqref{SufficiencyIntegralForQ} to first establish injectivity of the
operator $Q_i\mapsto\mathring{\Delta}Q_i+\mathring{r}_{ij}Q^j$, and
again appealing to the stability property of kernels of elliptic
operators.}
\end{remark}

\subsection{Parametrising the space of freely-prescribed data}
\label{Section:ParametrisingFreeData}

We have seen that, according to Theorem \ref{MainTheorem}, there exist
solutions of the extended constraints corresponding to
freely-prescribed data $(\phi,\bmT,\bar{\bmT})$ sufficiently close to
$(\bm0,\bm0,\bm0)$, where $\bmT,\bar{\bmT}\in
\mathscr{S}_{TT}(\mathcal{S};\mathring{\bmh})$. In this last subsection we aim to give an
explicit parametrisation of the space of freely-prescribed data, using
the ideas of \cite{Bei97} for the construction of transverse-tracefree tensors on
conformally flat manifolds, which have previously been applied to the
construction of  \emph{generalised} Bowen-York data ---see
\cite{Bei00}. We first review the basic ideas.

%Let $\mathcal{B}_{ijk}$ denote Cotton-York tensor:
%\[\mathcal{B}_{ijk}=2 D_{i}l_{j]k}, \]
%where $l_{ij}\equiv r_{ij}-\tfrac{1}{4}r h_{ij}$ is the \tetxit{Schouten curvature}. 

\subsubsection{The Gasqui--Goldschmidt complex}

Let $\mathcal{H}(\bmh)_{ij}$ denote the \emph{Cotton--York tensor} associated to
a metric $\bmh$ ---namely
\[
\mathcal{H}_{ij}\equiv \epsilon_{kl(i}D^k r_{j)}{}^l.
\]
% \mnotex{You don't strictly need to symmetrise - the antisymmetric part
% is precisely the contracted Bianchi identity, but it makes things
% slightly more straightforward when we linearise. (Otherwise you'd have
% to use the linearised contracted Bianchi identity to get symmetry.}
%where $l_{ij}\equiv r_{ij}-\tfrac{1}{4}r h_{ij}$, the \textit{Schouten tensor}.  
The Cotton tensor $\mathcal{H}_{ij}$ is symmetric and
tracefree. Moreover, by the \emph{third Bianchi identity} it is also
divergence-free. Recall also that, in dimension $3$, the vanishing of
the Cotton-York tensor is equivalent to local conformal-flatness
---see e.g. \cite{CFEBook}. Now consider the linearisation,
$\mathring{H}(\bmeta)_{ij}$,  about a background metric $\mathring{\bmh}$, given by the Fr\'{e}chet derivative
\begin{align*}
\mathring{H}(\bmeta)_{ij}&\equiv \frac{d}{d\tau}\mathcal{H}(\mathring{\bmh}+\tau\bm\eta)_{ij}\bigg\vert_{\tau=0} \\
%&=\mathring{\epsilon}^{kl}{}_i(\mathring{D}_k \breve{r}(\gamma)_{lj}-C(\bm\eta)^m{}_{kj}\mathring{r}_{lm})+\tfrac{1} %{2}\gamma\mathcal{H}_{ij}-2\mathring{\epsilon}_i{}^{[k}{}_j\eta^{l]j}D_{k} r_{lj}
&=\mathring{\epsilon}^{kl}{}_i(\mathring{D}_k \breve{r}(\eta)_{lj}-C(\bm\eta)^m{}_{kj}\mathring{r}_{lm})+\eta_{(i}{}^k\mathring{\mathcal{H}}_{j)k}-\tfrac{1}{2}\eta\mathring{\mathcal{H}}_{ij}
\end{align*}
with indices raised using $\mathring{\bmh}$. Here, $\eta\equiv
\text{tr}_{\mathring{\bmh}}(\bm\eta)$, the operator
$C(\cdot)^i{}_{jk}$is as defined
in \eqref{PerturbationOfGradient} and $\breve{r}(\eta)_{ij}$ is the
linearised Ricci operator acting on the metric perturbation
$\eta_{ij}$, and given by equation \eqref{LinearisedRicci}. 

\medskip
According to the above observations, if $\mathring{\bmh}$ is
\emph{conformally flat}, then $\mathring{H}(\bmeta)\in
\mathscr{S}^2_0(\mathcal{S};\mathring{\bmh})$. Moreover, in the case
of conformally-flat data, $H(\bmeta)_{ij}$ is also divergence-free 
since the linearisation of the third Bianchi identity gives 
\begin{align*}
0&=\frac{d}{d\tau}\delta_\bmh(\mathcal{H}(\bmh))_i \bigg\vert_{\tau=0}\\
&= \mathring{\delta}(H(\bmeta))_i-\eta^{kj}\mathring{D}_k\mathring{\mathcal{H}}_{ij}-\tfrac{1}{2}\mathring{\mathcal{H}}^{jk}\mathring{D}_i\eta_{jk}-\mathring{\mathcal{H}}_i{}^k\mathring{D}^j\eta_{jk}+\tfrac{1}{2}\mathring{\mathcal{H}}_i{}^k\mathring{D}_k\eta\\
&=\mathring{\delta}(\mathring{H}(\bmeta))_i
\end{align*}
where to pass from the second to the third line it has been used that
$\mathring{\mathcal{H}}_{ij}=0$ for a conformally flat background. Hence,
$\mathring{H}(\bmeta)_{ij}\in\mathscr{S}_{TT}(\mathcal{S};\mathring{\bmh})$. The
above features are expressed succinctly in the
\textit{Gasqui-Goldschmidt} elliptic complex ---see \cite{GasGol84,Bei97}:
\[
0\rightarrow \Gamma(\Lambda^1(\mathcal{S}))
\xrightarrow{\mathring{L}}\Gamma(\mathscr{S}_0^2(\mathcal{S};\mathring{\bmh}))\xrightarrow{\mathring{H}}\Gamma(\mathscr{S}_0^2(\mathcal{S};\mathring{\bmh}))\xrightarrow{\mathring{\delta}}\Gamma(\Lambda^1(\mathcal{S}))\rightarrow
0,
\]
which holds for any conformally flat manifold
$(\mathcal{S},\mathring{\bmh})$. Here, we are using $\Gamma(\cdot)$ to denote smooth sections of the indicated tensor bundle. Another consequence of the elliptic complex is that the linear sixth-order operator $P\equiv \mathring{H}^2+(\mathring{L}\circ\mathring{\delta})^3$ is elliptic ---see \cite{Bei97}. It is straightforward to see that $\text{ker}~P=\text{ker}~\mathring{H}\cap\text{ker}~\mathring{\delta}$, and hence that $P$ is injective for a conformally rigid manifold $(\mathcal{S},\mathring{\bmh})$.  
\medskip 

For compact
$\mathcal{S}$, the above elliptic complex admits the following expression of \emph{Poincar\'{e}
duality}:
\[
\text{ker}~\mathring{\delta}/\mathring{H}(\Gamma(\mathscr{S}^2_0(\mathcal{S};\mathring{\bmh})))\simeq\text{ker}~\mathring{H}/\mathring{L}(\Gamma(\Lambda^1(\mathcal{S}))). 
\]
Hence, \emph{given our assumption of conformal rigidity}, it follows
that the map 
\[
\mathring{H}:\Gamma(\mathscr{S}_0^2(\mathcal{S};\mathring{\bmh}))\rightarrow\Gamma(
\mathscr{S}_{TT}(\mathcal{S};\mathring{\bmh}))
\]
 is, in fact, surjective
---any \emph{smooth} TT tensor may be constructed as the image under $H$ of some 
smooth tracefree $2-$tensor. This result is generalised in the following Proposition:

\begin{proposition}
\label{Prop:SurjectivityOfH}
Let $(\mathcal{S},\mathring{\bmh})$ be a smooth conformally-rigid (not necessarily hyperbolic) manifold, then the map
\[ \mathring{H}:H^{s+2}(\mathscr{S}^2_0(\mathcal{S};\mathring{\bmh}))\rightarrow H^{s-1}(\mathscr{S}_{TT}(\mathcal{S};\mathring{\bmh})),\]
is surjective for $s\geq 4$.
\end{proposition}
\begin{proof}
Given $T_{ij}\in H^{s-1}(\mathscr{S}_{TT}(\mathcal{S};\mathring{\bmh}))$, then since $\Gamma(\mathscr{S}_{TT}(\mathcal{S};\mathring{\bmh}))\cap H^{s-1}(\mathscr{S}_{TT}(\mathcal{S};\mathring{\bmh}))$ is dense in $H^{s-1}(\mathscr{S}_{TT}(\mathcal{S};\mathring{\bmh}))$ we can approximate $T_{ij}$ by a Cauchy sequence $T^{(n)}_{ij}\in\Gamma(\mathscr{S}_{TT}(\mathcal{S};\mathring{\bmh}))$. Since $\mathring{\bmh}$ is conformally rigid there exists, for each $n\in\mathbb{N}$, an element $\eta^{(n)}_{ij}\in\Gamma(\mathscr{S}^2_0(\mathcal{S};\mathring{\bmh}))$ for which $\mathring{H}(\bm\eta^{(n)})_{ij}=T^{(n)}_{ij}$. Without loss of generality, we may assume that $\eta^{(n)}_{ij}\in\Gamma(\mathscr{S}_{TT}(\mathcal{S};\mathring{\bmh}))$ for each $n\in\mathbb{N}$ ---one takes the TT part of the York split of a given $\eta^{(n)}_{ij}$, if necessary, and uses the fact that $\text{Im}~\mathring{L}\subset \text{ker}~\mathring{H}$. Now since the elliptic operator $P\equiv \mathring{H}^2+(\mathring{L}\circ\mathring{\delta})^3$ is injective, it follows from standard results of elliptic PDE theory (see Appendix H of \cite{Bes08}, for instance) that there exists some constant $C>0$ for which the elliptic estimate 
\[\Vert \bm\eta\Vert_{H^{s+2}}\leq C\Vert P(\bm\eta)\Vert_{H^{s-4}}\]  
holds for all $\eta_{ij}\in H^{s+2}(\mathscr{S}^2_0(\mathcal{S};\mathring{\bmh}))$. In particular, it follows that
\begin{align*}
\Vert \bm\eta^{(m)}-\bm\eta^{(n)}\Vert_{H^{s+2}}& \leq C\Vert P(\bm\eta^{(m)}-\bm\eta^{(n)})\Vert_{H^{s-4}}\\
&\leq C\Vert \mathring{H}\circ \mathring{H}(\bm\eta^{(m)}- \bm\eta^{(n)})\Vert_{H^{s-4}}\\
&\leq C\Vert \mathring{H}(\bmT^{(m)}-\bmT^{(n)})\Vert_{H^{s-4}}\\
&\leq C\Vert \bmT^{(m)}-\bmT^{(n)}\Vert_{H^{s-1}}, 
\end{align*}
where the second line follows from the fact that, by assumption, $\eta^{(n)}_{ij}$ are divergence-free, and the fourth follows by continuity of $\mathring{H}$ as a map from $H^{s-1}$ to $H^{s-4}$. It follows that the sequence $\lbrace\bm\eta^{(n)}\rbrace$, $n\in\mathbb{N}$, is Cauchy in the $H^{s+2}$-norm and therefore converges to some $\eta_{ij}\in H^{s+2}(\mathscr{S}_{TT}(\mathcal{S};\mathring{\bmh}))$. By continuity we then have that $\mathring{H}(\bm\eta)_{ij}=T_{ij}$, as required. 
\end{proof}

\subsubsection{The parametrisation}

The above ideas can now be applied to obtain the \emph{parametrisation} of
the free data $T_{ij},~\bar{T}_{ij}$: 

\begin{proposition}
\label{Prop:ParametrisingFreeData}
Let $(\mathcal{S},\mathring{\bmh},\mathring{\bmK})$ satisfy the
conditions of Theorem \ref{MainTheorem}, and let $\mathcal{U}$ be the
neighbourhood of the freely specifiable data as
given there. There exists an open subset
\[
\tilde{\mathcal{U}}\subset\mathcal{B}_\bmeta\equiv H^{s-1}(\mathscr{S}^2_0(\mathcal{S};\mathring{\bmh}))\big),
\]
 such that:
\begin{enumerate}[i)]
\item for each $\bmeta,\bar{\bmeta}\in\tilde{\mathcal{U}}$ there exists a solution to the extended constraint equations with free data 
\begin{equation}
T_{ij}=\mathring{H}(\bmeta)_{ij}, \qquad \bar{T}_{ij}=\mathring{H}(\bar{\bmeta})_{ij}; \label{ParametrisationOfFreeData}
\end{equation}
\item all admissible free data (i.e. $\bmT,\bar{\bmT}\in\mathcal{U}$)
may be obtained in the form \eqref{ParametrisationOfFreeData}, for
some $\bmeta,\bar{\bmeta}\in\tilde{\mathcal{U}}$.
\end{enumerate} 
For a given $T_{ij}~\bar{T}_{ij}$, the choice of
$\eta_{ij},\bar{\eta}_{ij}$ in \eqref{ParametrisationOfFreeData} is
unique up to the addition of elements in $\text{Im}(\mathring{L})$.
\end{proposition}

\begin{proof}
Take $\tilde{\mathcal{U}}\equiv \mathring{H}^{-1}(\mathcal{U}\cap\text{Im}(\mathring{H}))$. The map
\[
\mathring{H}: \mathcal{B}_\bmeta\rightarrow\mathcal{B}_T 
\]
is continuous, so $\tilde{\mathcal{U}}$ is open in
$\mathcal{B}_T $. Applying Theorem \ref{MainTheorem} with free data
\eqref{ParametrisationOfFreeData} establishes \emph{(i)}. By assumption of conformal rigidity 
and using Proposition \ref{Prop:SurjectivityOfH} it follows that
\[ \mathring{H}:H^{s+2}(\mathscr{S}^2_0(\mathcal{S};\mathring{\bmh}))\rightarrow H^{s-1}(\mathscr{S}_{TT}(\mathcal{S};\mathring{\bmh}))\]
is surjective, so $\mathring{H}(\tilde{\mathcal{U}})=\mathcal{U}$, establishing \emph{(ii)}. Uniqueness
(up to addition of elements in $\text{Im}(\mathring{L})$) follows immediately from the assumption of conformal rigidity.
\end{proof}

\section{Conclusions and Outlook}

The Friedrich-Butscher method originally applied in \cite{But06,But07}
to the asymptotically flat case, was implemented here to the case of
hyperbolic background initial data. This method provides a promising alternative
to the standard conformal method for the construction of initial data;
in particular, it allows for the possibility of generating solutions
to the constraint equations that are tailored in the sense of having
certain components of the Weyl curvature (restricted to $\mathcal{S}$)
prescribed from the outset.

Work is currently under progress to extend the present results to a
broader class of background initial data, in addition to extending the
analysis to the full conformal constraint equations. It would be
interesting to see whether the method can be implemented numerically
through an iterative convergence scheme.

\section*{Acknowledgements}
The authors thank the hospitality of the International Erwin
Schr\"odinger Institute for Mathematics and Physics where a big part
of this work was carried out as part of the research programme
\emph{Geometry and Relativity} during
July-September 2017.

%QM 
%\bibliographystyle{reporthack}
% Ludovica
%\bibliographystyle{/Users/Juan/Documents/tex/reporthack}

% Path in QM 
%\bibliography{Newgrbib}
% Path in Ludovica
%\bibliography{/Users/Juan/Documents/tex/Newgrbib}

\begin{thebibliography}{10}

\bibitem{AndMon04}
L.~Andersson \& V.~Moncrief,
\newblock {\em Future complete vacuum spacetimes},
\newblock in {\em The Einstein equations and the large scale behaviour of
  gravitational fields}, edited by P.~T. Chru\'{s}ciel \& H.~Friedrich, page
  299, Birkh\"auser, 2004.

\bibitem{Bar05}
R.~Bartnik,
\newblock {\em Phase Space for the Einstein Equations},
\newblock Comm. Anal. Geom. {\bf 13}, 845 (2005).

\bibitem{BarIse04}
R.~Bartnik \& J.~Isenberg,
\newblock {\em The constraint equations},
\newblock in {\em The Einstein equations and the large scale behaviour of
  gravitational fields}, edited by P.~T. Chru\'{s}ciel \& H.~Friedrich, page~1,
  Birkhauser, 2004.

\bibitem{Bei97}
R.~Beig,
\newblock {\em TT-tensors and conformally flat structures on 3-manifolds},
\newblock Banach Center Publications {\bf 41}, 109 (1997).

\bibitem{Bei00}
R.~Beig,
\newblock {\em Generalized Bowen-York initial data},
\newblock in {\em Mathematical and Quantum Aspects of Relativity and
  Cosmology}, edited by S.~Cotsakis \& R.~Beig, volume 537 of {\em Lecture
  Notes in Physics}, page~55, Springer, 2000.

\bibitem{BeiChrSch05}
R.~Beig, P.~T. Chru\'{s}ciel, \& R.~Schoen,
\newblock {\em KIDs are non-generic},
\newblock Ann. Henri Poincare {\bf 6}, 155 (2005).

\bibitem{Bes08}
A.~L. Besse,
\newblock {\em Einstein Manifolds},
\newblock Springer Verlag, 2008.

\bibitem{But06}
A.~Butscher,
\newblock {\em Exploring the conformal constraint equations},
\newblock in {\em The conformal structure of spacetime: Geometry, Analysis,
  Numerics}, edited by J.~Frauendiener \& H.~Friedrich, Lect. Notes. Phys.,
  page 195, 2002.

\bibitem{But07}
A.~Butscher,
\newblock {\em Perturbative solutions of the extended constraint equations in
  General Relativity},
\newblock Comm. Math. Phys. {\bf 272}, 1 (2007).

\bibitem{Can81}
M.~Cantor,
\newblock {\em Elliptic operators and the decomposition of tensor fields},
\newblock Bull. Am. Math. Soc. {\bf 5}, 235 (1981).

\bibitem{Cho08}
Y.~Choquet-Bruhat,
\newblock {\em General Relativity and the Einstein equations},
\newblock Oxford University Press, 2008.

\bibitem{ChoKno04}
B.~Chow \& D.~Knopf,
\newblock {\em The Ricci flow: an introduction}, volume Vol. 110,
\newblock American Mathematical Society, 2004.

\bibitem{Delay09}
E.~Delay,
\newblock {\em Perturbative solutions to the extended constant scalar curvature equations on asymptotically hyperbolic manifolds},
\newblock Proceedings of the American Mathematical Society {\bf 7}, 137 (2009).

\bibitem{DeT80}
D.~M. DeTurck,
\newblock {\em The equation of prescribed Ricci curvature},
\newblock Bull. Am. Math. Soc. {\bf 3}, 701 (1980).

\bibitem{Edw12}
C.~Edwards,
\newblock {\em Advanced calculus of several variables},
\newblock Courier Corporation, 2012.

\bibitem{Fri83}
H.~Friedrich,
\newblock {\em Cauchy problems for the conformal vacuum field equations in
  General Relativity},
\newblock Comm. Math. Phys. {\bf 91}, 445 (1983).

\bibitem{GalMiaSch15}
G.~J. Galloway, P.~Miao, \& R.~Schoen,
\newblock {\em Initial data and the Einstein constraint equations},
\newblock in {\em General Relativity and Gravitation: A centennial
  perspective}, edited by A.~Ashtekar, B.~K. Berger, J.~Isenberg, \& M.~A.~H.
  Mac{C}allum, Cambridge University Press, 2015.

\bibitem{GasGol84}
J.~Gasqui \& H.~Goldschmidt,
\newblock {\em D\'{e}formations infinit\'{e}simales des structures conformes
  plates}, volume~52 of {\em Progress in Mathematics},
\newblock Birkh\"auser, 1984.

\bibitem{HolMaxMaz17}
M.~Holst, D.~Maxwell, \& R.~Mazzeo,
\newblock {\em Conformal Fields and the Structure of the Space of Solutions of
  the Einstein Constraint Equations},
\newblock in {\tt arXiv 1711.01042}.

\bibitem{Kap94}
M.~Kapovich,
\newblock {\em Deformations of representations of discrete subgroups of $SO (3,
  1)$},
\newblock Math. Ann. {\bf 299}, 341 (1994).

\bibitem{Laf83}
J.~Lafontaine,
\newblock {\em Modules de structures conformes plates et cohomologie de groupes
  discrets},
\newblock CR Acad. Sci. Paris Ser. I Math {\bf 297}, 655 (1983).

\bibitem{LawMic89}
H.~B. Lawson \& M.~L. Michelson,
\newblock {\em Spin geometry},
\newblock Princeton University Press, 1989.

\bibitem{Max11}
D.~Maxwell,
\newblock {\em A model problem for conformal parameterizations of the Einstein
  constraint equations},
\newblock Comm. Math. Phys. {\bf 302}, 697 (2011).

\bibitem{Max14}
D.~Maxwell,
\newblock {\em The conformal method and the conformal thin-sandwich method are
  the same},
\newblock Class. Quantum Grav. {\bf 31}, 145006 (2014).

\bibitem{Nak03}
M.~Nakahara,
\newblock{\em Geometry, topology and physics},
\newblock CRC Press, 2003. 

\bibitem{Ren08}
A.~D. Rendall,
\newblock {\em Partial differential equations in General Relativity},
\newblock Oxford University Press, 2008.

\bibitem{CFEBook}
J.~A. {Valiente Kroon},
\newblock {\em Conformal Methods in General Relativity},
\newblock Cambridge University Press, 2016.

\bibitem{Yor73}
J.~W. {York Jr}.,
\newblock {\em Conformally covariant orthogonal decomposition of symmetric
  tensor on Riemannian manifolds and the initial value problem of General
  Relativity},
\newblock J. Math. Phys. {\bf 14}, 456 (1973).


\end{thebibliography}

\end{document}